\documentclass{article}

\def\matches{\rightarrowtriangle}
\def\nat{\mathbb{N}}

\newcommand{\return}[1]{\widehat{#1}}
\newcommand{\rete}{\return{e}}

\newcommand{\refeq}[1]{(\ref{#1})}

\newcommand{\refprop}[1]{Proposition~\ref{#1}}
\newcommand{\reffig}[1]{Fig.~\ref{#1}}

\newcommand{\reflem}[1]{Lem\-ma~\ref{#1}}

\newcommand{\refsec}[1]{Section~\ref{#1}}
\newcommand{\refex}[1]{Example~\ref{#1}}
\newcommand{\refdef}[1]{Definition~\ref{#1}}

\newcommand{\reftab}[1]{Table~\ref{#1}}
\newcommand{\refasm}[1]{Assumption~\ref{#1}}
\newcommand{\refalg}[1]{Algorithm~\ref{#1}}

\usepackage{fullpage}

\usepackage{array, multirow}

\usepackage{stmaryrd}
\usepackage{epsfig}
\usepackage[usenames,dvipsnames,svgnames,table]{xcolor}
\usepackage{amssymb}
\usepackage{amsmath}
\usepackage{amsthm}

\usepackage{pstricks}

\usepackage[all]{xy}

\usepackage{algorithmic}
\usepackage{algorithm}

\usepackage[all]{xy}

\usepackage{xspace}

\newcommand{\trs}{{\mathcal L_R}}
\newcommand{\trd}{{\mathcal L_U}}

\newcommand{\ete}[1]{\xi(#1)}

\newcommand{\prei}[2]{#1 \lhd #2}
\newcommand{\preii}[3]{#1 \lhd #2 \rhd #3}

\newcommand{\trdi}[2]{{\sf Z}_{#2}(#1)}
\newcommand{\trdii}[3]{{\sf Z}_{#2}(#1,#3)}

\newcommand{\trsi}[2]{{\sf Y}_{#2}(#1)}

\newcommand{\esti}[4]{#1
  \stackrel{#3}{\longrightarrow}_{#2} #4} 
\newcommand{\estii}[4]{#1
  \stackrel{#3}{\Longrightarrow}_{#2} #4} 

\newcommand{\cl}{{\mathcal L}_U}
\newcommand{\clI}{{\mathcal L}_I}
\newcommand{\cls}{{\mathcal L}_R}

\newcommand{\pdelta}{{\mathcal Q}_{\delta}}
\newcommand{\sdelta}{{\mathcal S}_{\delta}}

\newcommand{\pcpi}[3]{#1_{#2}^{#3}}
\newcommand{\pcpr}[3]{\widehat{#1}{\,}_{#2}^{#3}}

\newcommand{\calA}{\mathcal A}
\newcommand{\calB}{\mathcal B}
\newcommand{\calM}{\mathcal M}

\newcommand{\calS}{\mathcal S}
\newcommand{\calQ}{\mathcal Q}

\newcommand{\machine}[1]{{\mathcal M} [#1]}
\newcommand{\machinea}[1]{{\mathcal M}_U [#1]}

\newcommand{\machinex}[1]{{\mathcal M'} [#1]}
\newcommand{\machineb}[1]{{\mathcal M}_R [#1]}

\newcommand{\machinexx}[1]{{\mathcal M''} [#1]}
\newcommand{\machinec}[1]{{\mathcal M}_V [#1]}

\newcommand{\machinepcp}[1]{{\mathcal Q}  [#1]}

\newcommand{\qsc}{quiescent sequential consistency\xspace}
\newcommand{\Qsc}{Quiescent sequential consistency\xspace}
\newcommand{\qc}{quiescent consistency\xspace}
\newcommand{\Qc}{Quiescent consistency\xspace}
\newcommand{\qct}{quiescent consistent\xspace}

\newcommand{\trace}{run\xspace}
\newcommand{\traces}{runs\xspace}
\newcommand{\Trace}{Run\xspace}

\newcommand{\parikh}[1]{PI(#1)}

\newcommand{\robcomment}[1]{{\large\bf Comment: #1}}
\newcommand{\bdcomment}[1]{{\color{blue} \large BD: #1}}

\newtheorem{theorem}{Theorem}
\newtheorem{lemma}{Lemma}
\newtheorem{proposition}{Proposition}
\newtheorem{definition}{Definition}
\newtheorem{assumption}{Assumption}
\newtheorem{example}{Example}

\title{Decidability and Complexity for \\Quiescent Consistency and its
  Variations}

\author{Brijesh Dongol \qquad Robert M. Hierons\\[2mm]
  Department of Computer Science, \\
  Brunel University London,
  UK\\
  \texttt{firstname.lastname@brunel.ac.uk}
}

\date{}

\begin{document}

\maketitle

\begin{abstract}
Quiescent consistency is a notion of correctness for a concurrent
object that gives meaning to the object’s behaviours in quiescent
states, i.e., states in which none of the object’s operations are
being executed. Correctness of an implementation object is defined in
terms of a corresponding abstract specification. This gives rise to
two important verification questions: \emph{membership} (checking whether a
behaviour of the implementation is allowed by the specification) and
\emph{correctness} (checking whether all behaviours of the implementation are
allowed by the specification). In this paper, we show that the
membership problem for quiescent consistency is NP-complete and that
the correctness problem is decidable, but coNP-hard and in
EXPSPACE. For both problems, we consider restricted versions of
quiescent consistency by assuming an upper limit on the number of
events between two quiescent points. Here, we show that the membership
problem is in PTIME, whereas correctness is in PSPACE. 

Quiescent consistency does not guarantee sequential consistency, i.e.,
it allows operation calls by the same process to be reordered when
mapping to an abstract specification. Therefore, we also consider
quiescent sequential consistency, which strengthens quiescent
consistency with an additional sequential consistency condition. We
show that the unrestricted versions of membership and correctness are
NP-complete and undecidable, respectively. When by placing a limit on
the number of events between two quiescent points, membership is in
PTIME, while correctness is in PSPACE. Finally, we consider a version
of quiescent sequential consistency that places an upper limit on the
number of processes for every run of the implementation, and show that
the membership problem for quiescent sequential consistency with this
restriction is in PTIME. 

\end{abstract}

\section{Introduction}

Due to the possibility of parallel executions, correctness of an
operation of a concurrent object cannot be stated in terms of pre/post
conditions. Instead, correctness is expressed in terms of a history of
operation invocation/response events, capturing the interaction
between a concurrent object and its client. There are many notions of
correctness for safety \cite{HeSh08,DongolDGS15} --- relaxed notions
are more permissive, and hence, allow greater flexibility in an
object's design. Such flexibility is necessary in the presence of
observations such as Amdahl's Law and Gustafson's Law
\cite{Gustafson88}, which show that sequential bottlenecks within an
implementation must be reduced to improve performance \cite{Shavit11}.

This paper studies \qc \cite{Shavit11,DerrickDSTTW14} a relaxed notion of
correctness for concurrent objects, derived from a similar notion in
replicated databases \cite{Ellis77}, that gives meaning to an object
in its \emph{quiescent states}, i.e., states in which none of its
operations are currently executing. Here, correctness is defined by
mapping a concurrent object's history (with potentially overlapping
operation calls) to a sequential history of its corresponding
specification object (with no overlapping operation calls).
\begin{enumerate}
\item A history of a concurrent object is considered to be correct
  with respect to a correctness condition $C$ iff the history can be
  mapped to a valid (sequential) history of the object's specification
  and the mapping satisfies $C$.
\item A concurrent object satisfies $C$ iff each of its histories is
  correct with respect to $C$.
\end{enumerate}
These two issues give rise to two distinct verification problems: the
former gives rise to a \emph{membership} problem, and the latter a
\emph{correctness} problem.  In this paper, we extend the existing
approach of Alur et al. \cite{AlurMP00} and study the
decidability and complexity of both membership and correctness for two correctness conditions: \qc and \qsc.  

Informally speaking, \qc is defined as follows. A concurrent object is
said to be in a quiescent state if none of its operations are being
executed in that state. \Qc allows operations calls in a concurrent
history between two consecutive quiescent states to be reordered when
mapping to history of the sequential specification, but disallows
events that are separated by a quiescent state from being reordered
\cite{DongolDGS15,DerrickDSTTW14,Shavit11}. Compared to other
conditions in the literature, \qc is more permissive. For example,
unlike linearizability \cite{HeWi90,HeSh08,DongolDGS15}, it allows the
effects of operation calls to be reordered even if they do not overlap
in a concurrent history. Unlike sequential consistency
\cite{Lamport79,HeSh08,DongolDGS15}, it allows the effects of
operation calls by the same process to be reordered.

In the context of client-object systems, to guarantee observational
refinement \cite{HeHS86} of the client, it turns out that it is
necessary ensure process order is maintained
\cite{Filipovic-LinvsRef2010}. Therefore, we additionally consider
\qsc, a variation of \qc that adds a \emph{sequential consistency}
constraint \cite{Lamport79} to \qc, i.e., we are not allowed to
reorder the events of the same process.

In this paper, we make the following main contributions.
\begin{enumerate}
\item We describe how \qc can be expressed using independence from
  Mazurkiewicz Trace Theory \cite{Mazurkiewicz84} and encoded as
  finite automata.
\item Show that deciding membership for \qc is an NP-complete problem
  if the number of events between two quiescent states is unbounded,
  but deciding membership is polynomial (with respect to the size of
  the input \trace) if the number of events between two quiescent
  states has a fixed upper bound.
\item Show that correctness for \qc is decidable, coNP-hard, and in
  EXPSPACE, but correctness for \qc is in PSPACE if the number of
  events between two quiescent states has a fixed upper bound.
\item Show that deciding membership for \qsc is an NP-complete problem
  but can be solved in polynomial time (with respect to the size of
  the input \trace) if the number of events between two quiescent
  states has a fixed upper limit, or if the number of processes can be
  predetermined. 
\item Show that correctness for \qsc is undecidable but is in PSPACE
  if the number of events between two quiescent states has a fixed
  upper bound.
\end{enumerate}

This paper is organised as follows. In \refsec{sec:background}, we
motivate the problem through an example, and describe the formal background of
finite automata and independence used in the rest of the
paper. \refsec{sec:quiesc-cons-1} defines \qc, develops a finite
automata encoding of \qc as well as the membership and correctness
problems. Our results for the membership and correctness problems are
given in \refsec{sec:membership-problem} and
\refsec{sec:correctness-problem}, respectively. \refsec{sec:qsc}
describes \qsc, then Sections~\ref{sec:membership-qsc} and
\ref{sec:correctness-qsc} present the results for membership and
correctness for \qsc, respectively. A survey of related work and
concluding remarks is given in \refsec{sec:conclusions}.


\section{Background}
\label{sec:background}
This section motivates \qc with a queue example
(\refsec{sec:quiesc-cons-queue}), then gives a finite automata
formalisation for studying the problem (\refsec{sec:probl-repr}). We
will use a notion of independence from Mazurkiewicz Trace Theory,
which we describe in \refsec{sec:independence}.

\subsection{A quiescent consistent queue}
\label{sec:quiesc-cons-queue}

We consider the \qct queue from \cite{DerrickDSTTW14}
(see Figs.~\ref{fig:diffqueue} and \ref{fig:diffqueue-code}). The
queue is based on the architecture of {\em diffracting trees}, which
uses the following principle (adapted from counting networks
\cite{AspnesHS1994}). Elements called {\em balancers} are arranged in
a binary tree, which may have arbitrary depth. Each balancer contains
one bit, which determines the direction in which the tree is
traversed; a balancer value of 0 causes a traversal up and a value 1
causes a traversal down.  The leaves of the tree point to a concurrent
data structure. Operations on the tree (and hence data structures)
start at the root of the tree and traverse the tree based on the
balancer values. Each traversal is coupled with a bit flip, so that
the next traversal occurs along the other branch. Upon reaching a
leaf, the process performs a corresponding operation on the data
structure at the leaf.

Our example consists of two 1-level balancers {\tt eb} and {\tt db}
used by enqueue and dequeue operations, respectively. Both operations
share the two queues at the leaves (see \reffig{fig:diffqueue}).
Pseudocode for the queue is given in \reffig{fig:diffqueue-code}. Both
operations are implemented using a non-blocking atomic {\tt CAS}
(Compare-And-Swap) operation that compares the stored local value
\texttt{old} with the shared variable \texttt{var} and updates \texttt{var}
to a new value \texttt{new} if the values of \texttt{var} and \texttt{old}
are still equal: \medskip

\noindent \hfill
\begin{minipage}[t]{0.9\columnwidth}
  \tt CAS(var,old,new) == 

  \qquad \ atomic\{ if var = old 

  \qquad \ \ \ \ \ \ \ \ \ then
  var := new; return true 
  
  \qquad \ \ \ \ \ \ \ \ \ else return false\}
\end{minipage}\hfill{}\medskip

\noindent Both operations read their corresponding bit and try to flip
it using a {\tt CAS}. If they succeed, they perform an {\it enqueue} {\tt
  Enq} or {\it dequeue} {\tt Deq} on the queue of their local bit. For
simplicity, we assume that {\tt Enq} and {\tt Deq} are atomic
operations (though they could be implemented by any linearizable
operation). The queue only satisfies \qc if {\tt Deq} is
\emph{blocking}, i.e., waits until an element is found in the
queue. The diffracting queue is not \qct if {\tt Deq}
returns on empty (see \cite{DerrickDSTTW14} for details).

\begin{figure}[t]
  \centering
  \begin{minipage}[b]{0.3\columnwidth}
    \centering
    \scalebox{0.7}{\input{diffqueue.pspdftex}}
    \caption{A 1-level diffracting queue with
      two queues}
    \label{fig:diffqueue}
  \end{minipage}
  \qquad\qquad 
  \begin{minipage}[b]{0.55\columnwidth}
    {\tt \footnotesize Init: eb, db = 0}

    \begin{minipage}[t]{0.49\columnwidth}
      \tt\footnotesize

        enqueue(el:T) 

        E1: do lb := eb; 

        E2: until CAS(eb,lb,1-lb) 

        E3: Enq(queue[lb],el) 
    \end{minipage}
    \hfill
    \begin{minipage}[t]{0.47\columnwidth}
      \footnotesize\tt
      
        dequeue

        D1: do lb := db;

        D2: until CAS(db,lb,1-lb)

        D3: return  Deq(queue[lb])
    \end{minipage}
    \caption{Enqueue and dequeue operations on the diffraction queue}
    \label{fig:diffqueue-code}
  \end{minipage}
\end{figure}

\begin{example}
  \label{ex:qch1}
  The following is a possible history for the blocking concurrent
  queue implementation: \smallskip


  \noindent \begin{small}
    \hfill$h_1 = D_1\ E_2(a)\ \return{E}_2\ E_3(b)\ \return{E}_3\ D_4\
    \return{D}_4(b)\ D_5\ \return{D}_5(a)\ E_6(c)\ \return{E}_6 \
    \return{D}_1(c)$\hfill{}
  \end{small}\smallskip

  \noindent 
  where $D_1$ denotes a {\tt dequeue} invocation by process $1$,
  $\return{D}_1(c)$ denotes a {\tt dequeue} by process $1$ that
  returns $c$, $E_2(a)$ denotes an {\tt enqueue} invocation by process
  $2$ with input $a$, and $\return{E}_2$ denotes the corresponding return
  event. 
  There is not much concurrency in $h_1$. Only the first
  dequeue is running concurrently with the rest of the
  operations. However, due to the first dequeue invocation, $h_1$ is
  only quiescent at the beginning and end.
  
  History $h_1$ is not linearizable \cite{HeWi90} because the dequeues
  by processes $4$ and $5$ violate the FIFO order of enqueues by
  processes $2$ and $3$, and linearizability does not allow
  non-overlapping operations to be reordered (see
  \cite{DerrickDSTTW14} for details). However, $h_1$ is \qct because
  \qc allows operations between two consecutive quiescent states to be
  reordered even if they do not overlap. This means that it may be
  matched with the following sequential history, which satisfies a
  specification of a sequential queue data structure. \smallskip
  
    \noindent\hfill$
    h_2 = E_3(b)\ \return{E}_3\ E_2(a)\ \return{E}_2\ D_4\
    \return{D}_4(b)\ D_5\ \return{D}_5(a)\ E_6(c)\ \return{E}_6 \ D_1\
    \return{D}_1(c) $\hfill{}
\qed
\end{example}

\subsection{Problem representation}
\label{sec:probl-repr}
In this section, we present our formal framework.
The behaviour of a system will be a sequence of events.
Given a set $A$ we will let $A^*$ denote the set of finite sequences
of elements of $A$ and $\varepsilon \notin A$ denote the empty
sequence.
Like Alur et
al. \cite{AlurMP00}, the specification and implementation are both
represented by finite automata, whose alphabet is a set of events
recording the invocation/response of an
operation. 

\begin{definition}
  A \emph{finite automaton} (FA) is a tuple $(M,m_0,\Sigma,t,M_\dagger)$ in
  which $M$ is the finite set of states, $m_0 \in M$ is the initial
  state, $\Sigma$ is the finite alphabet, $t: M \times \Sigma
  \leftrightarrow M$ is the transition relation and $M_\dagger \subseteq M$
  is the set of final states.
\end{definition}

\noindent 
Given a finite automaton $\calM = (M,m_0,\Sigma,t,M_\dagger)$, $m' \in
t(m,e)$ is interpreted as ``it is possible for $\calM$ to move from
state $m$ to state $m'$ via event $e$'' and this defines the
\emph{transition} $(m,e,m')$.  A \emph{path} of $\calM$ is a sequence
$\rho = (m_1,e_1,m_2), (m_2,e_2,m_3), \ldots, $ $(m_{k},e_k,m_{k+1})$
of consecutive transitions.  The path $\rho$ has starting state
$start(\rho) = m_1$, ending state $end(\rho) = m_{k+1}$ and label
$label(\rho) = e_1 e_2 \ldots e_k$.  We let $Paths(\calM)$ denote the set
of paths of $\calM$.  The FA $\calM$ defines the regular language $L(\calM)$ of
labels of paths that start in $m_0$ and end in final states.  More
formally,
\[
L(\calM) = \left\{
\begin{array}[c]{@{}l@{}}
label(\rho) \mid
\rho \in Paths(\calM) \wedge 
start(\rho) = m_0 \wedge end(\rho) \in M_\dagger
\end{array}\right\}
\]
Given $\sigma \in L(\calM)$ we let $\calM(\sigma)$ denote the
set of states of $\calM$ that are ending states of paths in
$Path(\calM)$ that have label $\sigma$. Note that because \qc is a
safety property, considering only finite \traces is adequate; this
restriction is also made by Alur et al. for linearizability
\cite{AlurMP00}.

Given an FA $\calM$ that represents either a specification or
implementation, $\Sigma$ (the alphabet of $\calM$) is the set of events,
and so, the language $L(\calM)$ denotes the possible sequences of events
(called \emph{\traces}). In this setting, each $\sigma \in L(\calM)$ of an
automaton representing an object is also a possible history of the object.

We will use $\calS = (S,s_0,\Sigma,t_S,S_\dagger)$ to denote the FA
that represents the specification and $\calQ =
(Q,q_0,\Sigma,t_Q,Q_\dagger)$ to denote the FA that represents the
implementation.  We will typically use $s_1, \ldots$ for the names of
states of $\calS$ and $q_1, \ldots$ for the names of states of
$\calQ$.  If $\calS$ is the FA for a sequential queue object, it will
generates \traces like $h_2$ in \refex{ex:qch1}, and if $\calQ$ is the
FA for the implementation in \reffig{fig:diffqueue-code}, then it will
generate \traces such as $h_1$. 
In this paper
we will be interested in two different problems.
\begin{enumerate}
\item Deciding whether a \trace $\sigma \in L(\calQ)$ of the implementation
  is allowed by the specification $\calS$. \hfill (membership)
\item Deciding whether all \traces of $\calQ$ are allowed by the
  specification $\calS$ and thus whether $\calQ$ is a correct implementation
  of $\calS$. \hfill (correctness)
\end{enumerate}

To model concurrent operations, we assume that an operation has
separate invoke and return events. We will use natural numbers $\nat$
to identify processes and make the following assumption, which is a
common restriction used in the literature.

\begin{assumption}
  \label{asm:proc-bounded}
  The number of processes in the specification and implementation is bounded.
\end{assumption}
This assumption is implicitly met by the fact that we use FA $\calS$
and $\calQ$. Others have considered infinite-state systems in the
context of linearizability \cite{BouajjaniEEH13}. Here, dropping
\refasm{asm:proc-bounded} causes the correctness problem for
linearizability to become undecidable, whereas correctness for linearizability
with a bound on the number of processes is decidable
\cite{AlurMP00}. To recover decidability in the infinite case, one
must place restrictions on the algorithms under consideration, in
particular, linearizability is EXPSPACE-complete for implementations
with ``fixed'' linearization points \cite{BouajjaniEEH13} (see
\cite{HeSh08,DD15-csur} for examples of such implementations).

Each event in $\Sigma$ is associated with a process, an operation, and
an input or output value. Like \cite{AlurMP00}, our theory is data
independent in the sense that the input and output values are
ignored. We simply assume that the event sets of the specification and
implementation are equal, and hence, every input/output that is
possible for an event of the implementation is also possible for the
specification.  Given process $p$, $\Sigma(p)$ denotes the set of
events associated with $p$.  We write $e \rightarrowtriangle e'$ to
denote that $e$ \emph{matches} $e'$, i.e., $e$ is an invoke event and
$e'$ the corresponding response, which holds whenever the process and
operation corresponding to $e$ and $e'$ are the same. We let
$\pi_p(\sigma)$ denote the \trace that restricts $\sigma$ to events of
process $p$, which is defined by
\begin{eqnarray*}
  \pi_p(\varepsilon)  = \varepsilon \qquad\quad
  \pi_p(e \sigma)  =  {\sf if}\ e \in \Sigma(p)\ {\sf then}\ e
  \pi_p(\sigma)\ {\sf else}\ \pi_p(\sigma)
\end{eqnarray*}
The empty \trace $\varepsilon$ is \emph{sequential}. A non-empty \trace
$\sigma = e_0 \dots e_k$ is \emph{sequential} iff $e_0$ is an invoke
event, for each even $i < k$, $e_i \matches e_{i+1}$, and if $k$ is
even, $e_k$ is an invoke event.  $\sigma$ is \emph{legal} iff
for each process $p$, $\pi_p(\sigma)$ is sequential. Legality ensures
that each process calls at most one operation at a time. Furthermore,
legality is \emph{prefix closed}, i.e., if $\sigma$ is legal, then all
prefixes of $\sigma$ are legal.

As is common in the literature, we make the following assumption on each specification object, which
essentially means that its operations are atomic.
\begin{assumption}
  \label{asm:spec-seq}
  The specification $\calS$ is sequential.
\end{assumption}
Furthermore, as is common in the literature
\cite{AlurMP00,HeWi90,HeSh08,JagadeesanR14}, we ignore the behaviour
of clients that use the concurrent object in question, but assume that
each client process calls at most one operation of the object it uses
at a time (although different client threads may call concurrent
operations). This is captured by \refasm{asm:legal} below.
\begin{assumption}
  \label{asm:legal}
  All \traces of specification $\calS$ and implementation $\calQ$ are legal.
\end{assumption}

\begin{example}
  \label{ex:trace-example}
  Consider the history $h_1$ from \refex{ex:qch1}. We have that $D_1
  \matches \return{D}_1(c)$, $E_2(a) \matches \return{E}_2$,
  etc. Furthermore, $h_1$ is legal because $\pi_p(\tau)$ is sequential
  for each process $p$. \qed
\end{example}

\subsection{Independence}
\label{sec:independence}

In this paper, we study \qc and explore how it can be represented in terms of
\emph{independence} from Mazurkiewicz Trace Theory
\cite{Mazurkiewicz84}. Here, a symmetric independence relation $I
\subseteq \Sigma \times \Sigma$ is used to define equivalence classes
of \traces. If $(e,e') \in I$, then consecutive $e$ and $e'$ within a
\trace can be swapped. The independence relation defines a partial
commutation --- some pairs of elements commute, but there may also be
pairs that do not.  This leads to an equivalence relation $\sim_I$,
where $\sigma \sim_I \sigma'$ iff \trace $\sigma$ can be transformed
into \trace $\sigma'$ via a sequence of rewrites of the form $\sigma_1
e\,e' \sigma_2 \rightarrow_I \sigma_1 e'\,e \sigma_2$ for each $(e,e')
\in I$.

\begin{example}
  For $h_1$ and $h_2$ in \refex{ex:qch1}, if $I = \Sigma \times
  \Sigma$ then $h_1 \sim_I h_2$.\qed
\end{example}

Given a \trace $\sigma$ we will let $[\sigma]_I = \{\sigma' \mid \sigma
\rightarrow_I \sigma'\}$ denote the set of \traces that can be produced
from $\sigma$ using zero or more applications of the rewrite rules
defined by $I$.  We will let $\clI(\calM) = \cup_{\sigma \in L(\calM)}
[\sigma]_I$ denote the set of \traces that can be formed from those in
$\calM$ using rewrites based on $I$. 
We can now state membership and correctness as stated in
\refsec{sec:probl-repr} more precisely as follows.
\begin{enumerate}
\item Deciding whether $\sigma \in \clI(\calS)$ for a given $\sigma \in
  L(\calQ)$. \hfill  (membership)
\item Deciding whether $L(\calQ) \subseteq \clI(\calS)$. \hfill (correctness)
\end{enumerate}
N.B., the correctness problem is sometimes referred to as the \emph{model
  checking} problem.  In the next section we explore how problems
associated with \qc can be expressed in this manner, and will see that
this requires the FA that represent the specification and
implementation to be slightly adapted.


\section{Quiescent consistency}
\label{sec:quiesc-cons-1}

In this section we define \qc and explore its properties. In
\refsec{sec:quiescent-traces}, we define quiescent \traces and state a
number of properties that will be used in the rest of the paper, then
in \refsec{sec:quiesc-cons}, we present an adaptation of the FA from
the previous section to enable reasoning about membership and
correctness for \qc. In
\refsec{sec:allowable-behaviour}, we define \qc,
state the
membership and correctness problems in terms of the adapted FA.
Sections \ref{sec:membership-problem} and \ref{sec:correctness-problem}
then explore these problems.

\subsection{Quiescent \traces}
\label{sec:quiescent-traces}

We first define quiescent \traces and state some properties that we use
in the rest of this paper.  If $\sigma = \sigma_1 e \sigma_2$ and $e$
is an invocation event, we say $e$ is a \emph{pending invocation} if
for all $e' \in \sigma_2$, $e \not \matches e'$. A \trace $\sigma$ is
\emph{quiescent} if it does not contain any pending invocations.
Thus, if a legal \trace is quiescent then there is a one-to-one
correspondence between invoke and response events.  A path $\rho =
(q_0,e_1,q_1), (q_1,e_2,q_2), \ldots, (q_{k-1},e_k,q_k)$ is
\emph{quiescent} if $label(\rho)$ is quiescent.
\begin{example} 
  \Trace $h_1$ in \refex{ex:qch1} is quiescent, but the \trace $h_1\,
  E_1(x)\, D_3\, \return{E}_1$ is not because the invocation $D_3$ is
  pending. Note that quiescence does not guarantee legality, e.g.,
  \traces $\return{D}_2(\xi)$ and $D_1\, D_1\, \return{D}_2(\xi)\,
  \return{D}_2(\xi)$ are both quiescent, but neither is legal.\qed
\end{example}
The following result links quiescence and legality.
\begin{proposition}
  \label{prop:qu-legal}
  Suppose $\sigma = \sigma_1 \sigma_2 \dots \sigma_k$ is a legal and
  quiescent \trace, such that each $\sigma_i$ (for $1 \le i \le k$) is
  a quiescent \trace. Then for all $1 \le i \le k$, $\sigma_i$ is
  legal.
\end{proposition}
 \begin{proof}
   Suppose $\sigma$ is a legal quiescent \trace and $\sigma = \sigma_1
   \sigma_2 \dots \sigma_k$, where each $\sigma_j$ is quiescent. If
   $k=1$ then we are done so assume that $k >1$. Let $\sigma_i$ for $1
   \le i \le k$ be the first subsequence that is not legal, i.e., there
   exists a process $p$ such that $\pi_p(\sigma_i)$ is non-empty and
   not sequential.  Because legality is prefix closed, $\sigma' =
   \sigma_1 \dots \sigma_{i-1}$ must be legal. Moreover, for each
   process $q$, $\pi_q(\sigma')$ is either empty, or a non-empty
   sequential \trace ending with a return event. Thus for $p$, we have
   that $\pi_p(\sigma'\sigma_i)$ is not sequential, which contradicts
   the assumption that $\sigma$ is legal. 
 \end{proof}

We say that a \trace $\sigma$ is \emph{end-to-end quiescent} iff it is
quiescent and all non-empty proper prefixes of $\sigma$ are not
quiescent.
We write $\ete{\sigma}$ to denote $\sigma$ being end-to-end quiescent.
For example, the \trace $h_1$ in \refex{ex:qch1} is
end-to-end quiescent, and $h_2$ is quiescent but not end-to-end
quiescent. The next result states that any legal quiescent \trace can
be expressed as the concatenation of legal end-to-end quiescent
\traces.
\begin{proposition}
  \label{prop:end-to-end-qu}
  Suppose $\sigma$ is a legal quiescent \trace. Then $\sigma$ can be
  written in the form $\sigma_1 \sigma_2 \dots \sigma_k$ such that
  each $\sigma_i$ is a legal end-to-end quiescent \trace.
\end{proposition}
 \begin{proof}
   If $\sigma$ is end-to-end quiescent, we are done. Otherwise, there
   must exist $\sigma_1$ and $\sigma_2$, where $\sigma = \sigma_1
   \sigma_2$, such that $\sigma_1$ is legal and end-to-end quiescent,
   and $\sigma_2$ is legal and quiescent.  Because $\sigma_2$ is
   quiescent, it is possible to inductively apply the construction
   above, which completes the proof. 
 \end{proof}

\subsection{Distinguishing quiescence}
\label{sec:quiesc-cons}

We now develop an extension to the FA in \refsec{sec:probl-repr} to
facilitate reasoning about \qc in an automata-theoretic setting. 
\Qc is defined in terms of quiescent \traces and so we will consider the behaviour of
the implementation/specification to be its quiescent \traces.
This assumption
is stated formally in terms of FA as follows.

\begin{assumption}
  \label{asm:quiescent-path}
  A path of $\calS$ (and $\calQ$) starting from the initial state of
  $\calS$ (and $\calQ$) is quiescent iff it ends in a final state of
  $\calS$ (and $\calQ$).
\end{assumption}
By \refasm{asm:spec-seq}, $\calS$ is sequential, and hence,
distinguishing its quiescent states is straightforward.
The following proposition gives a sufficient condition under which
it is possible to partition the state set of
$\calQ$ into quiescent states and non-quiescent states
(and so it is straightforward to ensure that Assumption \ref{asm:quiescent-path} holds).

\begin{proposition}
  \label{prop:qui-init-path}
  Suppose that every path of $\calQ$ starting from $q_0$ is a
  prefix of a legal quiescent path of $\calQ$.  If $\rho$ is a
  quiescent path of $\calQ$ such that $start(\rho) = q_0$ and
  $end(\rho) = q$, then all paths of $\calQ$ starting from $q_0$ and
  ending in $q$ are quiescent.
\end{proposition}

 \begin{proof}
   The proof is by contradiction. Assume that there exist $\rho$,
   $\rho'$ and $q$ such that paths $\rho$ and $\rho'$ end at $q$,
   $\rho$ is quiescent and $\rho'$ is not quiescent.  Since $\rho'$ can
   be completed to form a quiescent path, 
   there must be
   a path $\rho''$ from $q$ such that $\rho'\rho''$ is quiescent.
   Further, $\rho''$ must contain more responses than invokes.  Thus,
   since $\rho$ is quiescent we can conclude that the path $\rho
   \rho''$ from the initial state of $\calQ$ has more response than
   invokes.  This provides a contradiction as required, since, by
   \refasm{asm:legal}, all histories of $\calQ$ are legal. 
 \end{proof}

Note that in the proof of \refprop{prop:qui-init-path}, it might be
necessary to invoke a new operation in order to complete the
non-quiescent path $\rho$ under consideration and reach a quiescent
state. For example, consider our diffraction queue in
\refsec{sec:quiesc-cons-queue}, where the {\tt dequeue} operation that
\emph{blocks} when the queue is empty. Suppose we have a path $\rho'$
such that
\[
label(\rho') = D_1\ D_2\ E_3(x)\ \return{D}_2(x)\ \return{E}_3
\]
It is not possible for $\rho'$ to reach a quiescent state by only
completing the pending invocations in $label(\rho')$ --- the only
pending invocation $D_1$ is blocked because the queue is
empty. However, it is possible to reach a quiescent state by following
a path where a new enqueue operation is invoked (by some process), and
this new operation along with the pending $D_1$ in $label(\rho')$ is
completed by adding matching returns. This observation does not
invalidate our results, which only requires that we identify the
quiescent states.

We now work towards a definition of allowable behaviours for \qc
(\refsec{sec:allowable-behaviour}), stated in terms of an independence
relation (\refsec{sec:independence}). In particular, by using a
special event $\delta \notin \Sigma$ that signifies quiescence, we aim
to use the \emph{universal independence relation}:
\[
U = \Sigma \times
\Sigma
\] which defines a partial commutation that allows \emph{all} events
different from $\delta$ in a \trace to commute. Thus, the alphabet of
the FA we use is extended to $\Sigma_{\delta} = \Sigma \cup
\{\delta\}$. Note that using $U$ as the independence relation means
that matching invocations and responses of the specification may also
be reordered when checking both membership and correctness. However,
as is standard in the literature, we have assumed that all \traces of
the implementation are legal (\refasm{asm:legal}), and hence, do not
generate \traces such that a response precedes an invocation, i.e.,
commutations of a response that is followed by a matching invocation will
never be used.


We now consider how we should add $\delta$ events to the FA $\calS$
(representing the specification) and $\calQ$ (representing the
implementation) by extending their transition relations, which results
in automata $\calS_\delta$ and $\calQ_\delta$.

First, consider the specification $\calS$.  One option is to insist
that a $\delta$ is included in a \trace \emph{whenever} a quiescent
state is reached.  However, if we apply this approach to the
specification, then the \traces of $\calS$ will all be of the form
$\delta e_1 \rete_1 \delta e_2 \rete_2 \delta \ldots$, i.e., a \trace
$\sigma$ of the implementation can only be equivalent to a \trace
$\sigma'$ of $\calS$ under the partial commutation defined by $U$ if
$\sigma = \sigma'$, which is not what is intended under \qc. This
situation is a result of applying the restriction --- that one can
only reorder between instances of quiescence --- to \traces of the
(sequential) specification; this restriction should only be applied to
\traces of the implementation.  Thus, we should not require a $\delta$
to appear in a \trace of $\calS$ whenever a quiescent state is
reached. Instead, we rewrite $\calS$ to form an FA $\sdelta$ so that
if $s$ is a quiescent state of $\calS$ (i.e., after each return event)
then there is a self-loop transition $(s,\delta,s)$ in $\sdelta$.
These are the only transitions of $\sdelta$ that have label
$\delta$. Overall, we construct $\calS_\delta$ such that we
\emph{allow} the inclusion of $\delta$ whenever a \trace of $\calS$
reaches a quiescent state.

Now consider the implementation $\calQ$. Here, we must insist that
there is a $\delta$ in a \trace of $\calQ$ whenever a quiescent state
is reached, therefore we rewrite $\calQ$ to form an FA $\pdelta$ such
that if $q$ is a quiescent state of $\calQ$ then all transitions that
leave $q$ in $\pdelta$ have label $\delta$.  These are the only
transitions of $\pdelta$ that have label $\delta$. In particular, for
each quiescent state $q$ of $\calQ$ we simply add a new state
$q_{\delta}$, make $q_{\delta}$ the initial state of all transitions
of $\calQ$ that leave $q$, and add the transition
$(q,\delta,q_{\delta})$. If $q$ is a final state of $\calQ$, i.e., $q
\in Q_\dagger$, we will make $q_{\delta}$ a final state of
$\calQ_\delta$ instead of $q$. Overall, we construct $\calQ_\delta$
such that we \emph{require} the inclusion of $\delta$ when $\calQ$
reaches a quiescent state.

The inclusion of $\delta$ in \traces of $\calS$ allows us to compare \traces
of $\calS$ and $\calQ$ (once rewritten based on independence relation $U$).
\begin{example}
  Returning to \traces $h_1$ and $h_2$ in \refex{ex:qch1}, there are
  many possible $\delta$ extensions of $h_2$ (which is a \trace of the
  specification), for example:
  \begin{eqnarray*}
    h_{2.1}^\delta & = & \delta\:E_3(b)\:\return{E}_3\:\delta\:E_2(a)\:
    \return{E}_2\:\delta\:D_4\: \return{D}_4(b)\:\delta\:
    D_5\:\return{D}_5(a)\:\delta\:E_6(c)\:\return{E}_6\:\delta\: D_1\:
    \return{D}_1(c)\:\delta \\
    h_{2.2}^\delta & = & \delta\:E_3(b)\:\return{E}_3\:E_2(a)\:
    \return{E}_2\:\delta\:D_4\: \return{D}_4(b)\:\delta\:
    D_5\:\return{D}_5(a)\:\delta\:E_6(c)\:\return{E}_6\: D_1\:
    \return{D}_1(c)\:\delta \\
    h_{2.3}^\delta & = & \delta\:E_3(b)\:\return{E}_3\:E_2(a)\:
    \return{E}_2\:D_4\: \return{D}_4(b)\:
    D_5\:\return{D}_5(a)\:E_6(c)\:\return{E}_6 \: D_1\:
    \return{D}_1(c)\:\delta 
  \end{eqnarray*}
  In contrast, there is exactly one $\delta$ extension of $h_1$,
  namely $\delta\ h_1\ \delta$. If $h_2$ had been a \trace of the
  implementation, then the only $\delta$ extension of $h_2$ is
  $h_{2.1}^\delta$.  \qed
\end{example}

In addition to adding $\delta$ to the \traces of $\calS$ and $\calQ$, we must
also reason about \traces with $\delta$ removed. To this end, we define
the following projection
\begin{eqnarray*}
  \pi_\Sigma(\varepsilon)  =  \varepsilon \qquad\qquad 
  \pi_\Sigma(e \sigma)  =  {\sf if}\ e \in \Sigma\ {\sf then}\ e
  \pi_\Sigma(\sigma)\ {\sf else}\ \pi_\Sigma(\sigma)
\end{eqnarray*}
Thus, for example $\pi_\Sigma(\delta h_{1} \delta) = h_1$ and
$\pi_\Sigma(h_{2.1}^\delta) = h_2$.

 \subsection{Allowable \qct behaviours}
 \label{sec:allowable-behaviour}

 In this section we formalise what it means for a \trace of $\calQ$ to
 be allowed by $\calS$, and this is stated in terms of the extended
 automaton $\pdelta$. Under \qc, \traces $\sigma$ and $\sigma'$ are
 equivalent if they have the same (multi-)sets of events between two
 consecutive occurrences of quiescence.  As a result, all elements in
 $\Sigma$ commute (we do not care about the relative order of these
 events) but nothing commutes with $\delta$.

 Under \qc a quiescent \trace $\sigma$ is allowed by specification $\calS$
 if $\sigma$ can be rewritten to form a \trace of $\calS$ by permuting
 events between consecutive quiescent points.  We thus obtain the
 following definition.
 \begin{definition}
   \label{def:quiesc-cons}
   Suppose $\sigma = \sigma_1 \sigma_2 \ldots \sigma_k$ is a legal
   quiescent \trace and each $\sigma_i$ is legal and end-to-end
   quiescent (N.B. it is always possible to write a quiescent $\sigma$
   in such a form due to \refprop{prop:end-to-end-qu}).  Then $\sigma$
   is \emph{allowed by $\calS$ under \qc} iff there exists a permutation
   $\sigma'_i \sim_U \sigma_i$ for each $1 \le i \le k$ such that
   $\sigma'_1 \sigma'_2 \ldots \sigma'_k \in L(\calS)$.
 \end{definition}

 We now define what it means for a \trace $\sigma \in L(\pdelta)$ to be
 allowed by a specification $\calS$ under \qc. 

 \begin{definition}
   \label{def:allowed-delta}
   \Trace $\sigma \in L(\pdelta)$ is \emph{allowed} by $\calS$ under \qc if
   $\pi_\Sigma(\sigma)$ is allowed by $\calS$ under \qc.
 \end{definition}

 We say $\sigma'$ is a \emph{legal permutation} of a legal \trace
 $\sigma$ iff $\sigma \sim_U \sigma'$ and $\sigma'$ is legal.
 \begin{proposition}
   \label{prop:quies-perm}
   If $\sigma$ is legal and quiescent, then any legal permutation of
   $\sigma$ is quiescent.
 \end{proposition}

 We can now express the membership and correctness problems in terms
 of $\pdelta$ and $\sdelta$, instead of between $\calQ$ and $\calS$ as
 done in \refsec{sec:independence}.

 \begin{lemma}[Membership]\label{lem:def_membership}
   Suppose $\sigma \in L(\pdelta)$. Then $\sigma$ is allowed by $\calS$
   under \qc iff $\sigma \in \cl(\sdelta)$.
 \end{lemma}

  \begin{proof} Suppose $\sigma \in L(\pdelta)$. By Proposition
    \ref{prop:end-to-end-qu} and the construction of $\pdelta$, we have
    $\sigma = \delta \sigma_1 \delta \ldots \sigma_k \delta$ such that
    the $\sigma_i$ do not include $\delta$ (i.e., each $\sigma_i$ is
    end-to-end quiescent).

    First assume that $\sigma$ is allowed by $\calS$ under \qc. By
    \refdef{def:allowed-delta}, $\sigma_1 \sigma_2 \dots \sigma_k$ is
    allowed by $\calS$, and hence, by \refdef{def:quiesc-cons} we have that
    $\calS$ has a \trace $\sigma'_1 \sigma'_2 \ldots \sigma'_k$ such that
    $\sigma'_i$ is a legal permutation of $\sigma_i$ (for all $1 \le i
    \le k$).  Furthermore, $\calS$ is initially quiescent and by
    \refprop{prop:quies-perm}, each $\sigma_i'$ is quiescent, therefore
    $\sdelta$ has the \trace $\sigma ' = \delta \sigma'_1 \delta
    \sigma'_2 \delta \ldots \delta \sigma'_k\delta$.  By definition,
    $\sigma \in \cl(\sdelta)$ as required.

    Now assume $\sigma \in \cl(\sdelta)$.  Then, $L(\sdelta)$
    contains a \trace $\sigma ' = \delta \sigma'_1 \delta \sigma'_2
    \delta \ldots \delta \sigma'_k\delta$ for some $\sigma'_1, \ldots,
    \sigma'_k$ such that $\sigma'_i$ is a permutation of $\sigma_i$ (all
    $1 \le i \le k$).  We therefore have that $L(S)$ contains a \trace
    $\sigma'_1 \ldots \sigma'_k$ such that $\sigma'_i$ is a permutation
    of $\sigma_i$ (all $1 \le i \le k$), and hence, have that $\sigma $
    is allowed by $\calS$ as required.
  \end{proof}

 \begin{lemma}[Correctness]
   \label{lem:correctness}
   Under \qc, $\calQ$ is a correct implementation of $\calS$ iff $L(\pdelta)
   \subseteq \cl(\sdelta)$.
 \end{lemma}
 \begin{proof}
  By \reflem{lem:def_membership} and the definition of \qc. 
 \end{proof}

\section{The Membership Problem}
\label{sec:membership-problem}

In this section we explore the following problem: given a
specification $\calS$ and \trace $\sigma \in L(\pdelta)$, do we have
that $\sigma \in \cl(\sdelta)$? We show that this question is in
general NP-complete (\refsec{sec:unbounded-qc}), but by assuming an
upper bound between occurrences of two quiescent states, the question
can be solved in polynomial time (\refsec{sec:bounded-qc}).

 \subsection{Unrestricted \qc}
 \label{sec:unbounded-qc}

 We first establish that the membership problem for \qc is indeed in
 NP.
 \begin{lemma}
   The membership problem for \qc is in NP.
 \end{lemma}

 \begin{proof}
   Given a \trace $\sigma \in L(\pdelta)$ and a specification $\calS$,
   a non-deterministic Turing machine can solve the membership problem
   of deciding whether $\sigma \in \cl(\sdelta)$ as follows.  First,
   the Turing machine guesses a \trace $\sigma'$ of $\sdelta$ with the
   same length as $\sigma$.  The Turing machine then guesses a
   permutation $\sigma''$ of $\sigma$ that is consistent with the
   independence relation $U$.  Finally, the Turing machine checks
   whether $\sigma' = \sigma''$.  This process takes polynomial time
   and hence, since a non-deterministic Turing machine can solve the
   membership problem in polynomial time, the problem is in NP.
 \end{proof}

 We now prove that this problem is NP-hard by showing how instances of
 the one-in-three SAT problem can be reduced to it.  An instance of the one-in-three
 SAT problem is defined by boolean variables $v_1, \ldots, v_k$ and
 clauses $C_1, \ldots, C_n$ where each clause is the disjunction of
 three literals (a literal is either a boolean variable or the negation
 of a boolean variable).  The one-in-three SAT problem is to decide
 whether there is an assignment to the boolean variables such that each
 clause contains \emph{exactly} one true literal and is known to be
 NP-complete \cite{Schaefer78}.\footnote{Note that the one-in-three SAT
   problem differs slightly from the more well-known 3SAT problem.}

 The construction in the proof of the result below takes an instance
 of the one-in-three SAT problem and constructs a specification
 $\calS$ that has $k+1$ `main' states $s_0, \ldots, s_k$ and for
 boolean variable $v_i$ it has two paths from $s_{i-1}$ to $s_{i}$:
 one path $\rho_i^T$ has a matching invocation/response pair $e_j,
 \rete_j$ for every clause $C_j$ that contains literal $v_i$ and the
 other path $\rho_i^F$ has a matching invocation/response pair
 $e_j,\rete_j$ for every clause $C_j$ that contains literal $\neg
 v_i$.  The relative order of the pairs of events in $\rho_i^F$ and
 $\rho_i^T$ will not matter.  A path from $s_0$ to $s_{k+1}$ is of the
 form $\rho_1^{B_1} \rho_2^{B_2} \ldots \rho_k^{B_k}$ for some $B_1,
 \ldots, B_k \in \{T,F\}$.  Furthermore, the number of times that the
 events $e_j$ and $\rete_j$ appear in the label of the path is the
 number of literals in clause $C_j$ that evaluate to true under this
 assignment of values to $v_1, \ldots, v_k$.  As a result, such a path
 contains $e_j$ and $\rete_j$ exactly once iff the assignment of
 $B_i$ to $v_i$ (all $1 \le i \le k$) leads to exactly one literal in
 $C_j$ evaluating to true.  Thus, there is a path from $q_0$ to $q_k$
 that contains each $e_j$ and $\rete_j$ exactly once iff there is a
 solution to this instance of the one-in-three SAT problem.

 \begin{example}
   \label{ex:reduction}
   Suppose we have four boolean variables $v_1, \ldots, v_4$ and
   clauses $C_1 = v_1 \vee v_2 \vee \neg v_3$, $C_2 = v_1 \vee \neg
   v_2 \vee v_4$, and $C_3 = v_2 \vee v_3 \vee \neg v_4$. This leads
   to the FA shown in Figure
   \ref{fig:example_construction_membership}. In this, for example,
   the label of $\rho_1^T$ is $e_1 \rete_1e_2\rete_2$ because $C_1$
   and $C_2$ both have literal $v_1$, and the label of $\rho_2^F$ is
   $e_2\rete_2$ since $C_2$ is the only clause that contains literal
   $\neg v_2$.  Consider now the path $\rho_1^T \rho_2^F \rho_3^F
   \rho_4^F$, which has label $e_1\rete_1e_2\rete_2 e_2\rete_2
   e_1\rete_1e_3\rete_3$.  The label of this path tells us that if we
   assign true to $v_1$ and false to each of $v_2, v_3, v_4$ then
   clause $C_1$ contains two true literals (since $e_1$ appears
   twice), $C_2$ contains two true literals (since $e_2$ appears
   twice), and $C_3$ contains one true literal (since $e_3$ appears
   once).  Thus, this assignment is not a solution to this instance of
   the one-in-three SAT problem because more than one clause of $C_1$
   and $C_2$ evaluates to true. \qed

   \begin{figure}[t]
     \centering
       $\UseTips
       \entrymodifiers = {} \xymatrix @+13mm { *++[o][F]{s_0}
         \ar@/^2.0pc/[r]^{label(\rho_1^T) = e_1\rete_1e_2\rete_2}
         \ar@/^-2.0pc/[r]_{label(\rho_1^F) = \epsilon} & *++[o][F]{s_1}
         \ar@/^2.0pc/[r]^{label(\rho_2^T) = e_1\rete_1e_3\rete_3}
         \ar@/^-2.0pc/[r]_{label(\rho_2^F) = e_2\rete_2} & *++[o][F]{s_2}
         \ar@/^2.0pc/[r]^{label(\rho_3^T) = e_3\rete_3} \ar@/^-2.0pc/[r]_{
           label(\rho_3^F) = e_1\rete_1} & *++[o][F]{s_3}
         \ar@/^2.0pc/[r]^{label(\rho_4^T) = e_2\rete_2}
         \ar@/^-2.0pc/[r]_{label(\rho_4^F) = e_3\rete_3} & *++[o][F]{s_4}
       }
       $
       \caption{Finite automaton for
         \refex{ex:reduction}} \label{fig:example_construction_membership}
       \vspace{-2mm}
   \end{figure}
 \end{example}

 Note that we are not asking whether the clauses can be satisfied, but
 whether they can be satisfied in a way that makes \emph{exactly} one
 literal of each true. This is equivalent to asking whether we can make
 the clauses true if they are stated in terms of \emph{isolating
   disjunction} $\dot\vee$, where
 \[
 \dot\vee(v_1,v_2, \dots, v_n) = \bigvee_{1\le i\le n} (v_i \land \bigwedge_{j\neq
   i} \neg v_j)
 \]
 \begin{example}
   Checking the assignment in \refex{ex:reduction}, is equivalent to
   checking for a satisfying assignment to the clauses $\dot C_1 =
   \dot\vee(v_1, v_2, \neg v_3)$, $\dot C_2 = \dot\vee (v_1, \neg v_2,
   v_4)$, and $\dot C_3 = \dot\vee (v_2, v_3, \neg v_4)$, where, for
   example,
   \[
   \dot C_1 = (v_1 \land \neg v_2 \land v_3) \vee (\neg v_1 \land v_2
   \land v_3) \vee (\neg v_1 \land \neg v_2 \land \neg v_3)
   \]
 \end{example}

 We now prove NP-hardness of the membership problem.  The proof
 essentially uses the finite automaton described above and the \trace
 $\sigma = e_1 \rete_1 \ldots e_n \rete_n$.  Additional events are
 included to allow these events to be reordered. In particular, we add
 an initial invocation $e_0$ and a final response $\rete_0$ in the
 \trace and so the implementation is only quiescent in its initial
 state and at the end of the \trace.  This allows the events of the
 \trace to be reordered; without the initial invocation and final
 response we could only compare $\sigma$ with \traces of the
 specification in which the pairs $e_i, \rete_i$ are met in the order
 found in $\sigma$.

 \begin{lemma}
   \label{lem:membership-unrestricted-qc}
 The membership problem for \qc is NP-hard.
 \end{lemma}

 \begin{proof}
   Assume that we are given an instance of the one-in-three SAT problem
   defined by boolean variables $v_1, \ldots, v_k$ and clauses $C_1,
   \ldots, C_n$.  We define a specification with invocation events
   $e_0, e_1, \ldots, e_n, e$ and corresponding return events $\rete_0,
   \rete_1, \ldots, \rete_n, \rete$.  Define a finite automaton
   specification $\calS$ as follows.  The state set of $\sdelta$ includes
   states $s_0, s_1, \ldots, s_k$ and $s, s'$ with $s$ being the
   initial state.  For all $1 \le i \le k$ there are two paths from
   $s_{i-1}$ to $s_{i}$: path $\rho_i^T$ has invocation/response pair
   $e_j,\rete_j$ for every clause $C_j$ that contains literal $v_i$;
   and path $\rho_i^F$ has invocation/response pair $e_j,\rete_j$ for
   every clause $C_j$ that contains literal $\neg v_i$.  Thus, a path
   from $s_0$ to $s_{k}$ is of the form $\rho_1^{B_1} \rho_2^{B_2}
   \ldots \rho_k^{B_k}$ for some $B_1, \ldots, B_k \in \{T,F\}$.  From
   the initial state $s$ the path to $s_0$ has \trace $e_0 \rete_0$ and
   from $s_k$ the path to the final state $s'$ has \trace $e\, \rete$.

   Consider the \trace $\sigma = e_0e_1\rete_1\ldots
   e_{n}\rete_ne\,\rete\,\rete_0$. We prove that $\sigma$ is in
   $\cl(\sdelta)$ iff there is a solution to the instance of the
   one-in-three SAT problem defined by $v_1, \ldots, v_k$ and $C_1,
   \ldots, C_n$.  First note that if $\sigma \in \cl(\sdelta)$ then
   the corresponding \trace of $\sdelta$ must end at state $s'$ since
   $\sigma$ contains the events $e$ and $\rete$.  In addition,
   $\sigma$ is end-to-end quiescent and so we simply require that some
   permutation of $\sigma$ is in $L(\sdelta)$.  Thus, $\sigma \in
   \cl(\sdelta)$ iff $\sdelta$ has a path from $s_0$ to $s_k$ whose
   label $\sigma_1$ contains each $e_i$ and $\rete_i$ exactly once for
   $1 \le i \le n$.  Furthermore, $\sigma_1$ must be the label of a
   path of $\sdelta$ that is of the form $\rho = \rho_1^{B_1}
   \rho_2^{B_2} \ldots \rho_k^{B_k}$.  Thus, $\sigma \in \cl(\sdelta)$
   iff there is an assignment $v_1 = B_1, \ldots, v_k = B_k$ such that
   each clause $C_1, \ldots, C_n$ contains exactly one true literal.
   This is the case iff there is a solution to this instance of the
   one-in-three SAT problem.  The result now follows from the
   one-in-three SAT problem being NP-complete and the construction of
   $\sdelta$ and $\sigma$ taking polynomial time.
 \end{proof}

 The following brings together these results.

\begin{theorem}
  The membership problem for \qc is NP-complete.
\end{theorem}

\subsection{Upper bound for restricted \qc}
\label{sec:bounded-qc}

We now consider a restricted version of \qc that assumes an upper
limit on the number of events between two quiescent states.  It turns
out that the membership problem under this assumption is polynomial
with respect to the size of the specification $\calS$ and the length
of $\sigma$.  To prove this, we convert the membership problem into
the problem of deciding whether two finite automata define a common
word, which is a problem that can be solved in polynomial time. In
particular, for a given \trace $\sigma \in L(\calQ)$, we construct a
finite automaton $\machine \sigma$ (see \refdef{def:machine}) such
that $\sigma \in \cl(\sdelta)$ iff $L(\machine \sigma) \cap L(\calS)$
is non-empty.

 The proof (in particular, the construction of $\machinea \sigma$ in
 \refdef{def:MU}) requires that it is possible to distinguish each
 event within $\sigma$. Therefore, we introduce the following
 assumption. 
 
 \begin{assumption}
   For any $\sigma \in \calQ$, the events within $\sigma$ are unique,
   i.e., if $\sigma = e_1 \dots e_k$, for all $1 \le i < j \leq k$, we
   have $e_i \neq e_j$.
 \end{assumption}
 This assumption can be trivially ensured, for example, by labelling
 the events.  

 For any \trace $\sigma$ whose events are unique, we define a machine
 $\machinea \sigma$ that is a finite automaton for $\sigma$ that
 accepts any permutation of the events in $\sigma$.
 \begin{definition}
   \label{def:MU}
   Suppose $\sigma = e_1 \dots e_k$ is a \trace whose events are
   unique. We let $\machinea {\sigma}$ be the finite automaton $ (
   2^\Sigma, \emptyset, \Sigma, t, \{\Sigma\} ) $ where $\Sigma =
   \{e_1, \ldots, e_k\}$ and for all $T, T' \in 2^\Sigma$, we have
   $(T,e,T') \in t$ iff $e \not \in T$ and $T' = T \cup \{e\}$.
 \end{definition}
 Note that the construction $\machinea \sigma$ is generic, but we only
 use it in situations where $\sigma$ is legal and end-to-end quiescent.

 Next, we define $\machine \sigma$ for \traces $\sigma \in
 L(\pdelta)$. We use $L_1 \cdot L_2$ to denote the \emph{language
   product} of languages $L_1$ and $L_2$ and for FA $\calA$ and $\calB$, we let
 $\calA \cdot \calB$ be the FA such that $L(\calA \cdot \calB) = L(\calA) \cdot L(\calB)$. In
 what follows, $\calA$ only has one final state ($\calA$ is $\machinea {\sigma}$
 for some $\sigma$), and hence, we can construct $\calA \cdot \calB$ by adding
 an empty transition from the final state of $\calA$ to the initial state
 of $\calB$.

 \begin{definition}
   \label{def:machine}
   For \trace $\sigma = \delta \sigma_1 \delta \sigma_2 \delta \ldots
   \delta \sigma_k \delta $ such that $\ete{\sigma_i}$ for each $1
   \leq i \leq k$, we let $ \machine \sigma = \machinea {\sigma_1}
   \cdot \machinea {\sigma_2} \cdot \ldots \cdot \machinea {\sigma_k}.
   $
 \end{definition}

 The next result uses the automata construction in \refdef{def:machine}
 to convert the membership problem into a problem of deciding whether
 two automata accept a common word. Its proof is clear from the
 definitions.
 \begin{proposition}
   For any $\sigma \in L(\pdelta)$, we have $\sigma \in \cl(\sdelta)$ iff
   $L(\machine \sigma) \cap L(S) \neq \emptyset$.
 \end{proposition}

 We now arrive at our main result for this section. 
\begin{theorem}
  Suppose that there exists an upper limit $b \in \nat$, such that for
  each $\sigma \in L(\pdelta)$ there are at most $b$ events between
  two occurrences of $\delta$ in $\sigma$. Then the membership problem
  for \qc is in PTIME. 
\end{theorem}

 \begin{proof}
   By \refasm{asm:quiescent-path}, $\sigma$ is quiescent, and by
   \refprop{prop:end-to-end-qu} and the definition of $\pdelta$,
   $\sigma$ can be written as $ \sigma = \delta \sigma_1 \delta
   \sigma_2 \delta \ldots \delta \sigma_k \delta $, where each
   $\sigma_i$ is legal and end-to-end quiescent.

   For each $\sigma_i$, the size of $\machinea {\sigma_i}$ is
   exponential in terms of the length of $\sigma_i$.  If we place an
   upper limit $b$ on the number of events between two occurrences of
   quiescence then the size of $\machinea {\sigma_i}$ is polynomial
   (it is exponential in terms of $b$). Therefore, $\machine \sigma$
   is of polynomial size (the sum of the sizes of the $\machinea
   {\sigma_i}$) and the result follows from it being possible to
   decide whether $L(\machine \sigma) \cap L(S) \neq \emptyset$ in
   time that is polynomial in terms of the sizes of $\calS$ and
   $\machine \sigma$.
 \end{proof}

\section{The correctness problem}
\label{sec:correctness-problem}

For the correctness problem, we might directly compare $L(\calQ)$ and
$L(\calS)$, i.e., require that  $L(\calQ) \subseteq L(\calS)$.
However, this limits the potential for
concurrency --- $\calQ$ would essentially be sequential. The effect of
using a relaxed notions of correctness (such as \qc) is that it allows
$L(\calQ)$ to be compared with $L(\calS)$ using some notion of
\emph{observational equivalence}. Therefore, for \qc, we explore the
following problem: given an implementation $\calQ$ and specification
$\calS$, do we have that $L(\pdelta) \subseteq \cl(\sdelta)$?  We show
that this question is decidable, coNP-hard and in EXPSPACE.

A language is a \emph{rational trace language} if it is defined by a
finite automaton and a symmetric independence relation.  Decidability
of the correctness problem is proved by using the following result
from trace theory \cite{AalbersbergR89}.
 \begin{lemma}
   \label{lem:rat-tr-lang}
   Suppose $\calA$ and $\calB$ are FA with set of events $\Sigma$ and $I
   \subseteq \Sigma \times \Sigma$ is a symmetric independence
   relation. Then, the inclusion $\clI(\calA) \subseteq \clI(\calB)$ is
   decidable iff $I$ is transitive.
 \end{lemma}
 
 The following is an immediate consequence.

\begin{theorem}
  \label{thm:correctness-decidable}
  $L(\pdelta) \subseteq \cl(\sdelta)$ is decidable.
\end{theorem}

\begin{proof}
  The independence relation $U = \Sigma \times \Sigma$ is transitive.
  This result thus follows from Lemma \ref{lem:rat-tr-lang} and the
  fact that $L(\pdelta) \subseteq \cl(\sdelta)$ iff $\cl(\pdelta)
  \subseteq \cl(\sdelta)$.
\end{proof}

%



We now explore the complexity of the correctness problem, which is
equivalent to the complexity of deciding whether the inclusion
$\cl(\pdelta) \subseteq \cl(\sdelta)$ holds. We show that this problem
is coNP-hard by considering the problem of deciding inclusion of the
set of Parikh images of regular languages.  For the rest of this
section we assume that $\calA$ and $\calB$ are FA.

\subsection{Lower bound for unrestricted \qc}
 Given alphabet $\Sigma = \{e_1, \ldots, e_k\}$ and $\sigma \in
 \Sigma^*$, the \emph{Parikh image} of $\sigma$ is the tuple $(n_1,
 \ldots, n_k)$ such that $\sigma$ contains exactly $n_i$ instances of
 $e_i$ (all $1 \le i \le k$). 
 We use $\parikh \calA$ to denote the set of Parikh images of the \traces in $L(\calA)$
 and the inclusion problem for Parikh images is to decide whether
 $\parikh \calA \subseteq \parikh \calB$. Deciding inclusion for the Parikh
 images of regular languages is known to be
 coNP-hard (even if the size of the alphabets of both $\calA$ and $\calB$ are
 fixed) \cite{KopczynskiT10}. 

 To use the coNP-hard result for Parikh images, we construct FA
 $\calA'$ and $\calB'$ from $\calA$ and $\calB$ such that $\parikh
 \calA \subseteq \parikh \calB$ iff $\cl(\calA'_\delta) \subseteq
 \cl(\calB'_\delta)$, where $\calA'_\delta$ (and $\calB'_\delta$)
 extends $\calA'$ (resp. $\calB'$) with $\delta$ events and transitions
 as defined in \refsec{sec:quiesc-cons}.  
 Suppose $\Sigma$ is the alphabet of both $\calA$ and $\calB$.  For
 each $x \in \Sigma$ we define an invoke event $e_x$ and corresponding
 response event $\rete_x$. We also include an additional invoke event
 $e$ and corresponding response $\rete$ that do not correspond to any
 $x \in \Sigma$ and hence, the resulting event set is:
 \[\Gamma =
 \{e,\rete\} \cup \{e_x \mid x \in \Sigma\}\cup \{\rete_x \mid x\in \Sigma\}
 \]

 \noindent To construct FA $\calA'$, we initialise the state set of
 $A'$ to the state set of $A$ and the event set of $\calA'$ to
 $\Gamma$. We then modify $A'$ and construct the initial state,
 transitions, and final states of $\calA'$ as follows.
 \begin{enumerate}
 \item 
   For the initial state $q_0$ of $\calA$, add a new state $q'_0 \notin
   A$ to $A'$, make $q'_0$ the initial state of $\calA'$, and add
   the transition $(q'_0,e,q_0)$ to $\calA'$.
 \item For each transition $t = (q,x,q')$ in $\calA$, add transitions
   $(q,e_x,q_t)$ and $(q_t,\rete_x,q')$ in $\calA'$, where $q_t \notin
   A$, then add $q_t$ to $A'$.
 \item Add a state $q_F \notin A$ to $A'$, make this the only
   final state, and from every final state $q$ of $\calA$, add the
   transition $(q,\rete,q_F)$.
 \end{enumerate}
 We have the following relationship between $L(\calA)$ and $L(\calA')$.
 \begin{proposition}\label{prop:A-A'} 
   $x_1 x_2 \ldots x_k \in L(\calA)$ iff $e\,e_{x_1} \rete_{x_1}\allowbreak
   e_{x_2} \rete_{x_2} \allowbreak \ldots e_{x_k} \rete_{x_k} \rete \in
   L(\calA')$.
 \end{proposition}
 One important property of $\calA'$ is that every $\sigma \in L(\calA')$ is
 end-to-end quiescent. Thus, under \qc, 
 $\sigma$ is allowed by the specification of $\calA'$ iff some permutation
 of $\sigma$ is in the language defined by the specification.

 FA $\calB'$ is constructed as follows. Initialise the state set of $B'$ to the state set of $B$ and
 the event set of $\calB'$ to $\Gamma$, then set the initial state of $\calB$
 as the initial state of $\calB'$. Then perform the following.
 \begin{enumerate}
 \item For each transition $t = (q,x,q')$ in $\calB$, add transitions
   $(q,e_x,q_t)$ and $(q_t,\rete_x,q')$ to $\calB'$ for a new state $q_t
   \notin B'$, then add $q_t$ to $B'$.
 \item Add new states $q''$ and $q_F$ to $B'$, then for every final
   state $q$ of $\calB$ add transitions $(q,e,q'')$ and $ (q'',\rete,q_F)$
   to $\calB'$.  Finally, make $q_F$ the only final state of $\calB'$.
 \end{enumerate}
 We have the following relationship between $L(\calB)$ and $L(\calB')$.
 \begin{proposition}\label{prop:B-B'}
   $x_1x_2 \ldots x_k \in L(\calB)$ iff $e_{x_1} \rete_{x_1} \allowbreak
   e_{x_2} \rete_{x_2} \allowbreak \ldots e_{x_k} \rete_{x_k}
   \allowbreak e\, \rete \in L(\calB')$.
 \end{proposition}

 The next lemma links inclusion of Parikh images for $\calA$ and $\calB$ to
 inclusion of the languages of $\calA'_\delta$ and $\calB'_\delta$ under
 independence relation $U$. 
 
 \begin{lemma}\label{lemma_inclusion_Parikh}
   $\parikh \calA \subseteq \parikh \calB$ iff
   $\cl(\calA'_\delta) \subseteq \cl(\calB'_\delta)$.
 \end{lemma}

 \begin{proof}
   First assume $\parikh \calA \subseteq \parikh \calB$.  Suppose that
   $\sigma \in \cl(\calA'_{\delta})$; it is sufficient to prove that
   $\sigma \in \cl(\calB'_{\delta})$.  By \refprop{prop:A-A'} there is
   some $x_1x_2 \ldots x_k \in L(\calA)$ such that $\sigma \sim_U \delta
   \sigma' \delta$, where $\sigma' = e\, e_{x_1} \rete_{x_1}
   \allowbreak e_{x_2}\rete_{x_2} \allowbreak \ldots \allowbreak
   e_{x_k}\rete_{x_k} \allowbreak \rete$.  Since $\parikh \calA
   \subseteq \parikh \calB$ we have that $L(\calB)$ contains a
   permutation $y_1 \ldots y_k$ of $x_1 \ldots x_k$.  By
   \refprop{prop:B-B'}, $e_{y_1}\rete_{y_1} \allowbreak e_{y_2}
   \rete_{y_2} \allowbreak \ldots \allowbreak e_{y_k} \rete_{y_k}
   \allowbreak e\, \rete \in L(\calB')$ and so we also have that
   $\delta \sigma'' \delta \in \cl(\calB'_{\delta})$ where $\sigma'' =
   e_{y_1} \rete_{y_1} \allowbreak e_{y_2} \rete_{y_2} \allowbreak
   \ldots \allowbreak e_{y_k} \rete_{y_k} \allowbreak e\, \rete$. As
   $y_1 \ldots y_k$ is a permutation of $x_1 \ldots x_k$, $\sigma''
   \sim_U \sigma'$. Since $\delta \sigma'' \delta \in
   \cl(\calB'_{\delta})$ and $\sigma'' \sim_U \sigma'$ we have that
   $\delta \sigma' \delta \in \cl(\calB'_{\delta})$.  Thus, since
   $\sigma = \delta \sigma' \delta$, we have that $\sigma \in
   \cl(\calB'_{\delta})$ as required.

   Now assume $\cl(\calA'_{\delta}) \subseteq \cl(\calB'_{\delta})$. Suppose
   that $\gamma \in \parikh \calA$ and so there is some $\sigma' = x_1
   \ldots x_k$ in $L(\calA)$ with Parikh Image $\gamma$.  By
   \refprop{prop:A-A'}, $e\, e_{x_1} \rete_{x_1}\allowbreak e_{x_2}
   \rete_{x_2} \allowbreak \ldots \allowbreak e_{x_k}
   \rete_{x_k}\allowbreak \rete \in L(\calA')$.  Thus, $\delta\, e\,
   e_{x_1} \rete_{x_1} \allowbreak e_{x_2} \rete_{x_2}\allowbreak
   \ldots \allowbreak e_{x_k} \rete_{x_k}\allowbreak \rete\, \delta \in
   \cl(\calA'_{\delta})$.  Since $\cl(\calA'_{\delta}) \subseteq
   \cl(\calB'_{\delta})$, $\delta\, e\, e_{x_1} \rete_{x_1} \allowbreak
   e_{x_2} \rete_{x_2} \allowbreak \ldots\allowbreak e_{x_k}
   \rete_{x_k}\allowbreak \rete\, \delta \in \cl(\calB'_{\delta})$.  By
   construction, this implies that $e_{y_1} \rete_{y_1} \allowbreak
   e_{y_2} \rete_{y_2} \allowbreak \ldots \allowbreak e_{y_k}
   \rete_{y_k} \allowbreak e\, \rete \in L(\calB')$ for some permutation
   $y_1 \ldots y_k$ of $x_1 \ldots x_k$.  By \refprop{prop:B-B'} we
   therefore know that $y_1 \ldots y_k \in L(\calB)$.  Finally, since $y_1
   \ldots y_k$ and $x_1 \ldots x_k$ are permutations of one another
   they have the same Parikh Image and so $\gamma \in \parikh \calB$ as
   required. 
 \end{proof}

We therefore have the following result.

 \begin{theorem}
 The correctness problem for Quiescent Consistency is coNP-hard.
 \end{theorem}

 \begin{proof}
   By \reflem{lemma_inclusion_Parikh} and inclusion of Parikh images
   being coNP-hard. 
 \end{proof}

\subsection{Upper bound for unrestricted \qc}
\label{sec:upper-bound-unrestr}
We now investigate the upper bounds on the complexity of
deciding correctness of \qc and show that the problem is in
EXPSPACE. This proof is much more involved than the lower bound result
as it is necessary to first derive an algorithm for checking
correctness quiescent consistency (see \refalg{algorithm1}) and derive
an upper bound on its running time.  

We start by introducing some new notation. For $m \in M$ and FA $\calM
= (M,m_0,\Sigma,t,M_\dagger)$, we let $\prei{m}{\calM}$ denote the FA
$(M,m,\Sigma,t,M_\dagger)$ formed by replacing the initial state of
$\calM$ by $m$. Furthermore, for $M' \subseteq M$ (recalling that
$\ete{\sigma}$ denotes that $\sigma$ is end-to-end quiescent), we
define:
\begin{align*}
  \trdi{m}{\calM} & {} = \{ \sigma \in \cl(\prei{m}{\calM}) \mid
  \ete{\sigma}\} 
  & \qquad 
  \trdi{M'}{\calM} & {} = \bigcup_{m \in M'} \trdi{m}{\calM}
\end{align*}
Thus, $\trdi{m}{\calM}$ is the set of 
end-to-end quiescent \traces that start in state $m$ of $\calM$.  The
following is immediate from this definition.

\begin{proposition}
If $\calQ$ is a correct implementation of $\calS$ with respect
to \qc and $q_0$ and $s_0$ are the initial states of $\calQ$ and $\calS$
respectively then
$\trdi{q_0}{\calQ} \subseteq \trdi{s_0}{\calS}$.
\end{proposition}


We will use an implicit powerset construction when reasoning about
\qc. 
Given states $m, m' \in M$ of $\calM$, sets of states $M_1, M_2
\subseteq M$ and \trace $\sigma$, we define some further notation:
\begin{align*}
  \esti{m}{\calM}{\sigma}{m'} & \text{ iff }
  \begin{array}[t]{@{}l@{}}
  \exists \rho \in Paths(\calM) \, .\, start(\rho) = m \land end(\rho) = m'  \land  label(\rho) \in [\sigma]_U
  \end{array}
  \\
  \estii{M_1}{\calM}{\sigma}{M_2} & \text{ iff } 
  \forall \rho \in Paths(\calM) \, .\,  
  start(\rho) \in M_1 \land label(\rho) \in
  [\sigma]_U  
  \Rightarrow 
  end(\rho) \in M_2
\end{align*}
Thus, $\esti{m}{\calM}{\sigma}{m'}$ holds iff there is some path in
$\calM$ with labels in $[\sigma]_U$ from state $m$ to state
$m'$. Furthermore, $\estii{M_1}{\calM}{\sigma}{M_2}$ holds iff every
path of $\calM$ starting from a state in $M_1$ with label in
$[\sigma]_U$ ends in a state of $M_2$.



If $\calQ$ is not a correct implementation of $\calS$ with respect to
\qc then there must be a quiescent \trace $\sigma$ that demonstrates
this.  We will use the following result, which shows that if there is a counterexample to \qc
then there is one of the form $\sigma = \sigma_1 \dots \sigma_{k+1}$ (where $\xi(\sigma_i)$)
such that $\sigma_{k+1}$ is
the portion of $\sigma$ that is in $\calQ$ but not in $\calS$ (under
independence relation $U$) and $k$ is bounded by $|Q| \cdot
2^{|S|}$.

\begin{proposition}\label{prop:NT}
  $\calQ$ is not a correct implementation of $\calS$ under \qc iff
  there exists some \trace $\sigma = \sigma_1 \ldots \sigma_{k+1}$ for
  end-to-end quiescent $\sigma_1, \ldots, \sigma_{k+1}$ and
  corresponding pairs $(q_0,S_0), \allowbreak (q_1, S_1),\allowbreak
  \ldots,\allowbreak (q_{k},S_{k}) \in Q \times 2^S$ such that $S_0 =
  \{s_0\}$, $\esti{q_{i-1}}{\calQ}{\sigma_i}{q_i}$ and
  $\estii{S_{i-1}}{\calS}{\sigma_i}{S_i}$ (all $1 \leq i \leq k$) such
  that: 

\begin{enumerate}

\item $\sigma_{k+1} \in \trdi{q_k}{\calQ}$ and  $\sigma_{k+1} \not \in
  \trdi{S_k}{\calS}$, and

\item $k \leq |Q| \cdot 2^{|S|}$.

\end{enumerate}
\end{proposition}

\begin{proof}
  The existence of such a $\sigma$ demonstrates that $\calQ$ is not a
  correct implementation of $\calS$ under \qc and so it is sufficient
  to prove the left-to-right direction.  We therefore assume that
  $\calQ$ is not a correct implementation of $\calS$ under \qc.  Thus,
  there exists a quiescent \trace $\sigma$ that is in $L(\calQ)$ but
  not in $\trd(\calS)$.  Assume that we have a shortest such \trace
  $\sigma$, $\sigma = \sigma_1 \ldots \sigma_{k+1}$ for end-to-end
  quiescent $\sigma_1, \ldots, \sigma_{k+1}$.  Since $\sigma$ is in
  $L(\calQ)$ but not in $\trd(\calS)$, by the minimality of $\sigma$
  we must have that $\sigma_{k+1} \in \trdi{q_k}{\calQ}$ and
  $\sigma_{k+1} \not \in \trdi{S_k}{\calS}$ and so the first condition
  holds.  Further, by the minimality of $\sigma$ we must have that
  $(q_i,S_i) \neq (q_j,S_j)$, all $0 \leq i < j \leq k$; otherwise we
  can remove $\sigma_{i} \ldots \sigma_{j-1}$ from $\sigma$ to obtain
  a shorter \trace that is in $\trd(\calQ)$ but not in $\trd(\calS)$.
  But, there are $|Q| \cdot 2^{|S|}$ possible pairs and so the second
  condition, $k < |Q| \cdot 2^{|S|}$, must hold.  
\end{proof}

Using \refprop{prop:NT}, we develop \refalg{algorithm1}, which defines
a non-deterministic Turing Machine that solves the problem of deciding
correctness. At each iteration, the non-deterministic Turing Machine
first checks whether $\trdi{q_c}{\calQ} \not \subseteq
\trdi{S_c}{\calS}$; if not, it has demonstrated that $\calQ$ is not a
correct implementation of $\calS$ (the first condition of Proposition
\ref{prop:NT}).  If this condition holds then the non-deterministic
Turing Machine increments the counter $c$ and guesses a next pair
$(q_c,S_c)$.  It then checks that there is some $\sigma_c$ such that
$\esti{q_{c-1}}{\calQ}{\sigma_c}{q_c}$ and
$\estii{S_{c-1}}{\calS}{\sigma_c}{S_c}$.  If there is such a
$\sigma_c$ then the process can continue, otherwise the result is
inconclusive.  The bound on $c$ ensures that the algorithm terminates
as long as we can decide the conditions contained in the {\bf if}
statements (we explore this below).

\begin{algorithm}
\begin{algorithmic}
\STATE 
$c = 0$, $S_0 = \{s_0\}$, $Q_0 = \{q_0\}$

\WHILE{$c \leq  |Q| \cdot 2^{|S|}$}

\IF{$\trdi{q_c}{\calQ} \not \subseteq \trdi{S_c}{\calS}$} \STATE
Return Fail
\ENDIF

\STATE c = c+1

\STATE Choose some $(q_c,S_c) \in Q \times 2^S$

\IF{
  $\not\! \exists \sigma_c$ such that
  $\esti{q_{c-1}}{\calQ}{\sigma_c}{q_{c}}$ and
  $\estii{S_{c-1}}{\calS}{\sigma_c}{S_c}$}
\STATE Return Ok  
\ENDIF

\ENDWHILE
\end{algorithmic}
\caption{Deciding correctness for \qc}
\label{algorithm1}
\end{algorithm}

If a non-deterministic Turing Machine operates as above then it will
return Fail if there is some sequence of choices that leads to Fail
being returned.  The following is thus immediate from
\refprop{prop:NT}.

\begin{proposition}
If a non-deterministic Turing Machine applies Algorithm \ref{algorithm1} to $\calQ$ and $\calS$ then it
returns Fail iff $\calQ$ is not a correct implementation of $\calS$ with respect to \qc.
\end{proposition}

We now consider the two problems encoded in the conditions of
Algorithm \ref{algorithm1}: deciding whether $\trdi{q_c}{\calQ} \not
\subseteq \trdi{S_c}{\calS}$; and deciding whether there exists a
\trace $\sigma_c$ such that $\esti{q_{c-1}}{\calQ}{\sigma_c}{q_{c}}$
and $\estii{S_{c-1}}{\calS}{\sigma_c}{S_c}$.

We start with problem of deciding whether $\trdi{q_c}{\calQ} \subseteq
\trdi{S_c}{\calS}$.  This involves checking whether the Parikh Image
of one regular language is contained in the Parikh Image of another
regular language.  It is known that this problem can be solved in
non-deterministic exponential time (NEXPTIME) \cite{Huynh85}.

\begin{proposition}\label{QC_checkinclusion1}
It is possible to decide whether $\trdi{q_c}{\calQ} \subseteq \trdi{S_c}{\calS}$
in NEXPTIME.
\end{proposition}

The remaining problem we need to decide, for states $q_{c-1},q_c$ of
$\calQ$ and sets $S_{c-1}, S_c$ of states of $\calS$, whether there
exists some $\sigma_c$ that can
\begin{enumerate}\renewcommand{\labelenumi}{(\roman{enumi})}
\item take $\calQ$ from $q_{c-1}$ to $q_c$ and
\item take $\calS$ from the set $S_{c-1}$ of states to the set $S_c$
  of states.
\end{enumerate}

We introduce some further notation. For $m \in M$ and $M' \subseteq
M$, we let $\preii{m}{\calM}{M'}$ denote the FA
$(M,m,\Sigma,t,M')$ formed by making $m$ the initial state of
$\calM$ and $M'$ the final states.  We introduce the
following (assuming all states in $M'$ and $M''$ are quiescent).
\begin{align*}
  \trdii{m}{\calM}{M'} = {} & \left\{ \sigma \in
    \cl(\preii{m}{\calM}{M'}) \mid \xi(\sigma)\right\}
  & \qquad 
  \trdii{M'}{\calM}{M''} = {} & \bigcup_{m\in M'} \trdii{m}{\calM}{M''}
\end{align*}
That is, $\trdii{m}{\calM}{M''}$ is the set of end-to-end quiescent
\traces of $\calM$ that start in state $m$ and end at a state in $M''$. We
use shorthand $\trdii{m}{\calM}{m'}$ for $\trdii{m}{\calM}{\{m'\}}$
(similarly $\trdii{M'}{\calM}{m'}$).

Using this notation, condition (i) above may be formalised as the
predicate $\sigma_c \in \trdii{q_{c-1}}{\calQ}{q_c}$. Condition (ii)
above requires that $\sigma_c$ can take $\calS$ to all states in $S_c$
(and so that $\sigma_c \in \bigcap_{s \in S_c}
\trdii{S_{c-1}}{\calS}{s}$) and cannot take $\calS$ from $S_{c-1}$ to
any state outside of $S_c$ (and so that $\sigma_c \not \in \bigcup_{s
  \in (S \setminus S_c)} \trdii{S_{c-1}}{\calS}{s}$).  The negation of
the overall condition thus reduces to the following.
\begin{equation}
  \not \exists \sigma_c \in (A \setminus B ) \cap C\label{eq:1}
\end{equation}
where 
\begin{align*}
  A & {} =  \textstyle \bigcap_{s \in S_c} \trdii{S_{c-1}}{\calS}{s}
  &
  B & {} = \textstyle \bigcup_{s \in (S \setminus S_c)} \trdii{S_{c-1}}{\calS}{s} 
  &
  C & {} = \trdii{q_{c-1}}{\calQ}{q_c}
\end{align*}
Using some straightforward set manipulation, \refeq{eq:1} is equal to
$A \cap C \subseteq B$.











Thus, the problem is reduced to deciding whether the intersection of a
set of Parikh Images of regular languages is contained within the
Parikh Image of another regular language.  We also note that if we use
$\overline{L}$ to represent the complement of a language $L$ then $B
\subseteq C$ iff $B \subseteq C \cup \overline{A}$, and by de Morgan's
Law $\overline{\bigcap_i A_i}$ is equivalent to $\bigcup_i
\overline{A_i}$. The condition therefore becomes
\[
\begin{array}[t]{@{}l@{}}
\trdii{q_{c-1}}{\calQ}{q_c} \subseteq 
 \left(\textstyle \bigcup_{s \in (S \setminus S_c)}
\trdii{S_{c-1}}{\calS}{s}\right) \cup 
\left(\bigcup_{s \in S_c} \overline{\trdii{S_{c-1}}{\calS}{s}}\right)
\end{array}
\]


The Parikh Image of a regular language can be
represented by a semi-linear set that contains exponentially many
terms \cite{KopczynskiT10}.
In addition,
the complement of a semi-linear set can be represented by polynomially many terms \cite{Huynh80}.
Thus,
all of 
$\trdii{q_{c-1}}{\calQ}{q_c}$,
$\bigcup_{s \in (S \setminus S_c)} \trdii{S_{c-1}}{\calS}{s}$, 
and $\bigcup_{s \in S_c} \overline{\trdii{S_{c-1}}{\calS}{s}}$ can be represented using
exponentially many terms (linear sets).
Further,
the problem of deciding whether one semi-linear set is contained in
another is in $\Sigma_2^p$ \cite{Huynh86}\footnote{Cited in \cite{Huynh85}.} and so is in PSPACE.
The overall problem is thus in EXPSPACE (since there are exponentially many terms).

\begin{proposition}\label{QC_checkmove1}
It is possible to decide whether there exists \trace $\sigma_c$ such that
$\esti{q_{c-1}}{\calQ}{\sigma_c}{q_{c}}$ and $\estii{S_{c-1}}{\calS}{\sigma_c}{S_c}$
in EXPSPACE.
\end{proposition}

We can now bring these results together.


\begin{theorem}
The correctness problem for \qc is in EXPSPACE.
\end{theorem} 
\begin{proof}
  We know that a non-deterministic Turing Machine can use Algorithm
  \ref{algorithm1} to solve the problem.  Further, by Propositions
  \ref{QC_checkinclusion1} and \ref{QC_checkmove1} the conditions of
  the {\bf if} statements can be solved in NEXPTIME and EXPSPACE.
  Observe also that the storage required for the algorithm, beyond
  determining the conditions in the {\bf if} statements, is polynomial
  since the algorithm only has to store the current values of $q_c$,
  $S_c$ and $c$, the latter taking $\log(|Q|2^{|S|})$ space.  Since
  NEXPTIME is contained in EXPSPACE, we therefore have that a
  non-deterministic Turing Machine can solve the problem in
  nondeterministic EXPSPACE (NEXPSPACE).  The result now follows from
  Savitch's theorem \cite{Savitch}, which implies that NEXPSPACE =
  EXPSPACE.
\end{proof}

%

\subsection{Upper bound for restricted quiescent consistency}

We now consider the case where there is a limit $b$ on the lengths of
subsequences of \traces of $\calQ$ 
between two occurrences of quiescence.

\begin{proposition}\label{QC_checkinclusion2}
  If there is a bound on the length of end-to-end quiescent \traces
  in $\calQ$ and $\calS$, then it is possible to decide whether
  $\trdi{q_c}{\calQ} \subseteq \trdi{S_c}{\calS}$ in PSPACE.
\end{proposition}

\begin{proof}
  A nondeterministic Turing Machine can solve this problem in PSPACE
  as follows.  First, it guesses a \trace $\sigma$ whose length is at
  most the upper bound.  It then checks that $\sigma$ is end-to-end
  quiescent.  It then checks whether $\sigma \in \trdi{q_c}{\calQ}$
  and whether $\sigma \in \trdi{S_c}{\calS}$; we know that these
  checks can be performed in polynomial time since this is an instance of the
  restricted membership problem.  Finally, it returns
  failure if and only if $\sigma \in \trdi{q_c}{\calQ}$ and $\sigma
  \not \in \trdi{S_c}{\calS}$.
\end{proof}

\begin{proposition}\label{QC_checkmove2}
Let us suppose that there is a bound on the length of end-to-end quiescent \traces in $\calQ$ and $\calS$.
It is possible to decide whether there exists \trace $\sigma_c$ such that
$\esti{q_{c-1}}{\calQ}{\sigma_c}{q_{c}}$ and $\estii{S_{c-1}}{\calS}{\sigma_c}{S_c}$
in PSPACE.
\end{proposition}

\begin{proof}
  A nondeterministic Turing Machine can solve this problem in PSPACE
  as follows.  First, it guesses a \trace $\sigma$ whose length is at
  most the upper bound and checks that $\sigma$ is end-to-end
  quiescent. It then checks whether
  \begin{align*}
    \sigma \in \trdii{q_{c-1}}{\calQ}{q_{c}} \quad \text{and} \quad \sigma \in
    \trdii{S_{c-1}}{\calS}{S_c} \setminus \trdii{S_{c-1}}{\calS}{S \setminus
        S_c}.
  \end{align*}
  We know that the first check (solving the membership problem for
  bounded \qc) can be performed in polynomial time.  The second check
  can be solved by deciding whether $\sigma \in \trdii{S_{c-1}}{\calS}{S_c}$ and
  whether $\sigma \in \trdii{S_{c-1}}{\calS}{S \setminus S_c}$ and, again, these
  checks can be performed in polynomial time.  The nondeterministic
  Turing Machine returns True iff it finds that $\sigma \in
  \trdii{q_{c-1}}{\calQ}{q_c}$,  $\sigma \in \sigma \in \trdii{S_{c-1}}{\calS}{S_c}$, and $\sigma \not
  \in \trdii{S_{c-1}}{\calS}{S \setminus S_c}$.
%
%
\end{proof}

\begin{theorem}
The correctness problem for bounded \qc is in PSPACE.
\end{theorem}

\begin{proof}
From Propositions \ref{QC_checkinclusion2} and \ref{QC_checkmove2}
we know that the two conditions in Algorithm \ref{algorithm1} can
be decided in PSPACE.
Thus,
a nondeterministic Turing Machine
can apply Algorithm \ref{algorithm1} using polynomial space.
The result thus follows. 
\end{proof}

%

\section{\Qsc}
\label{sec:qsc}

In this section, we consider \qsc, which adds a \emph{sequential
  consistency} constraint \cite{Lamport79} to \qc, i.e., we are not
allowed to reorder the events of the same process. For concurrent
objects, this means that the order of effects of operation calls by
the same process will take place in program order: if operation calls
identified by events $e$, $\return{e}$ and $e'$, $\return{e}{\,}'$ all
have the same process, and a concrete implementation has a \trace
where $\return{e}$ occurs before $e'$, then such a trace cannot be
justified by a sequential \trace where $e'$ occurs before
$\return{e}$.

In the context of client-object systems, sequential consistency has
been shown to be equivalent to \emph{observational refinement}
\cite{HeHS86} provided that the client threads are
independent (i.e., do not share data)
\cite{Filipovic-LinvsRef2010}. Observational refinement provides the
conditions necessary for replacing a specification object within a
client program by an implementatation. The sorts of guarantees that
\qc provides a client is still a subject of further study; as Shavit
says, exploiting concurrency in the multiprocessor age requires a
rethinking of traditional notions of correctness \cite{Shavit11}.


We now present some background for \qsc in preparation for the
membership and correctness
problems. 
In order to formally define \qsc, we define a projection function that
also preserves $\delta$ states in the projection.

\begin{definition}
  Given \trace $\sigma \in \Sigma_{\delta}^*$, event $e \in
  \Sigma_{\delta}$ and process $p$, $\pi^\delta_p(\sigma)$ is defined
  by the following:
\begin{align*}
  \pi^\delta_p(\varepsilon) = \varepsilon \qquad\qquad 
  \pi^\delta_p(e \sigma) =
  {\sf if}\ e \in \Sigma(p) \cup \{\delta\}\ {\sf then}\ e
  \pi^\delta_p(\sigma)\ {\sf else }\ \pi^\delta_p(\sigma)
\end{align*}
For $\sigma, \sigma' \in \Sigma_{\delta}^*$, we write $\sigma \approx
\sigma'$ iff $\pi^\delta_p(\sigma) = \pi^\delta_p(\sigma')$ for every
process $p$.
\end{definition}

%
%

We can now define \qsc in a similar manner to \qc,
except that we include the constraint that events on a process are ordered.

\begin{definition}
  Suppose $\sigma = \sigma_1 \sigma_2 \ldots \sigma_m \in \Sigma^*$ is a quiescent
  \trace and each $\sigma_i$ is end-to-end
  quiescent. 
  Then $\sigma$ is \emph{allowed by specification $\calS$ under \qsc}
  iff there exists a permutation $\sigma'_i$ for each $1 \le i \le m$
  such that $\sigma_i \approx \sigma_i'$ and $\sigma'_1 \sigma'_2
  \ldots \sigma'_m$ is a \trace of $\calS$.
\end{definition}


Since we cannot reorder events on a process we obtain the following independence relation.
\[
R = \{(a,b) \mid \exists p, p'. p \neq p' \wedge a \in \Sigma(p) \wedge b \in \Sigma(p')\}
\]
The essential idea is that quiescence ($\delta$) does not commute with
anything, as with \qc, and that two events from $\Sigma$ are
independent if and only if they are on different processes.

Given an FA $M$ with alphabet $\Sigma_{\delta}$, we will use $\cls(M)$
to denote $\clI(M)$ in which the independence relation $I$ is $R$.  

We will use the FA $\pdelta$ and $\sdelta$, which are derived from the
implementation $\calQ$ and specification $\calS$, respectively, via
the construction described in \refsec{sec:quiesc-cons}.  If we
consider a quiescent \trace $\sigma$ of $\pdelta$ we have that
$\delta$ is included whenever $\sigma$ is quiescent.  We can define
what it means for a \trace that includes $\delta$ to be allowed by
$\calS$.

\begin{definition}\label{def_qsc_with_delta}
  \Trace $\sigma$ of $\pdelta$ is \emph{allowed} by $\calS$ under \qsc if
  the \trace $\pi_{\Sigma}(\sigma)$ formed from $\sigma$ by removing all instances of $\delta$
  is allowed by $\calS$ under \qsc.
\end{definition}

Recall also that all
processes observe quiescence.  As a result, we have the following
property.

\begin{lemma}
  Suppose $\sigma = \sigma_1 \sigma_2 \ldots \sigma_m$ is a quiescent
  \trace and each $\sigma_i$ is end-to-end
  quiescent.
  Given \trace $\sigma' \in \Sigma_{\delta}^*$ we have that $\sigma'
  \approx \sigma$ iff $\sigma' = \sigma'_1 \sigma'_2 \ldots
  \sigma'_m$ for some $\sigma'_1, \ldots, \sigma'_m$ with $\sigma_j
  \approx \sigma'_j$ (all $1 \leq j \leq m$).
\end{lemma}

Based on Definition \ref{def_qsc_with_delta}, this leads directly to
the following simplified ways of expressing when a \trace is allowed
under \qsc.

\begin{proposition}\label{prop:sufficient}
  Suppose $\sigma \in L(\pdelta)$ is a quiescent \trace. Then the following
  statements are equivalent:
  \begin{enumerate}
  \item $\sigma$ is allowed by $\calS$ under \qsc.
  \item There exists a $\sigma' \in L(\sdelta)$ such that $\sigma'
    \approx \sigma$.
  \item $\sigma  \in \trs(\sdelta)$.
  \end{enumerate}
\end{proposition}


The following lemma links \qc and \qsc.
\begin{lemma}
  If $\sigma \in L(\pdelta)$ is allowed by $\calS$ under \qsc then
  $\sigma$ is allowed by $\calS$ under \qc, but not
  vice-versa.
\end{lemma}
\begin{proof}
  The first part follows from the independence relation $R$ for \qsc
  being a subset of the independence relation $U$ for \qc.  To prove
  the second part it is sufficient to obtain a \trace $\sigma$ and
  specification $\calS$ such that $\sigma$ is allowed by $\calS$ for
  \qc but not for \qsc.  Suppose $\calQ$ allows \trace $\sigma = e\,
  e_1 e_2 \return{e}_1 \return{e}_2 \return{e}$ where events $e_1$,
  $e_2$, $\return{e}_1$ and $\return{e}_2$ are on the same process $p$
  and $\calS$ allows the \trace $\sigma' = e\, \return{e}\, e_1
  \return{e}_1 e_2 \return{e}_2$ but no other permutation of
  $\sigma$. Then $\sigma'$ is allowable under \qc, but not under \qsc.
\end{proof}


\section{Membership for \qsc}
\label{sec:membership-qsc}
In this section we consider the membership problem for \qsc. Some of
the results are similar to those for \qc, and hence, the proofs for
these results are elided.  The structure of this section is similar to
\refsec{sec:membership-problem} --- we first present the unrestricted
case (\refsec{sec:unrestricted-qsc}), then present the upper bounds for
the restricted cases (\refsec{sec:restricted-qsc}).

\subsection{Unrestricted \qsc}
\label{sec:unrestricted-qsc}

The unrestricted version of \qsc is NP-complete. First, we show that
the problem is in NP, then show that the problem is NP-hard.
\begin{lemma}
The membership problem for \qsc\ is in NP.
\end{lemma}

\begin{proof}
Given \trace $\sigma$ and specification $\calS$, a non-deterministic Turing machine
can solve the membership problem, of deciding whether $\sigma \in \cls(\sdelta)$,
as follows.
First, the Turing machine guesses a \trace $\sigma'$ of $\sdelta$ with the same length as $\sigma$.
The Turing machine then guesses a permutation $\sigma''$ of $\sigma$ that is consistent
with the independence relation $R$. 
Finally, the Turing machine checks whether $\sigma'' = \sigma'$.
This process takes polynomial time and so,
since a non-deterministic Turing machine can solve the membership problem in polynomial time,
the problem is in NP.
\end{proof}

We can adapt the proof, that the membership problem for \qc is NP-hard
(Lemma \ref{lem:membership-unrestricted-qc}), by simply having a
separate process for each invoke/response pair.  We therefore have the
following.

\begin{lemma}
The membership problem for \qsc is NP-hard.
\end{lemma}

\begin{theorem}
The membership problem for \qsc is NP-complete.
\end{theorem}

\subsection{Restricted \qsc} 
\label{sec:restricted-qsc}

We now adapt the approach developed for \qc, that showed that the
membership problem can be solved in polynomial time if we either have
have an upper limit on the number of events between any two instances
of quiescence, or on the number of processes in the system. 

\paragraph{Upper limit on number of events between quiescence}

Like \refdef{def:MU}, we construct a finite automaton that
accepts any permutation of \trace $\sigma$ that preserves the order of
events within a single process.  For a \trace $\sigma = e_1, \ldots,
e_k$ of distinct elements, $1 \leq i \leq k$ and process $p$, we let
\begin{align*}
  pre_p(\sigma, i) & = \{e_j \mid 1 \leq j < i \land e_j \in \Sigma(p)\}
\end{align*}
be the set of elements of $\sigma$ with index smaller than $i$ that
are part of process $p$.
\begin{definition}
  Suppose \trace $\sigma_i = e_1 \ldots e_k$ is such that for all $i, j
  \in \{1,\dots,k\}$, if $i \neq j$ then $e_i \neq e_j$. We let
  $\machineb {\sigma_i}$ be the finite automaton
  $(2^\Sigma,\emptyset,\Sigma,t,\{\Sigma\})$ such that:
  \begin{itemize}
  \item $\Sigma = \{e_1, \ldots, e_k\}$ and,
  \item for all $T, T' \in 2^\Sigma$ and $e_i \in \Sigma$, we have
    $(T,e_i,T') \in t$ iff $e_i \not \in T$, $T' = T \cup \{e_i\}$ and
    $pre_p(\sigma, i) \subseteq T$.
  \end{itemize}
\end{definition}


Using this definition, we obtain a new FA $\machinex \sigma$.

\begin{definition}
Given \trace $\sigma = \sigma_1 \sigma_2 \ldots \sigma_k \in \Sigma^*$
such that each $\sigma_i$ is end-to-end quiescent,
\[
\machinex \sigma = \machineb {\sigma_1} \cdot \machineb {\sigma_2} \cdot \ldots \cdot \machineb {\sigma_k}.
\]
\end{definition}

The following is clear from the definition and from Proposition \ref{prop:sufficient}.

\begin{lemma}
  Given \trace $\sigma$ and specification $\calS$, $\sigma \in
  \trs(\sdelta)$ iff $L(\machinex \sigma) \cap L(S) \neq \emptyset$.
\end{lemma}

As before, we have the following result.

\begin{theorem}\label{theorem:poly_memb_qsc1}
  If $b$ is an upper limit on the number of events between two
  occurrences of quiescence in each \trace of $\calQ$, then
  the membership problem for \qsc is in PTIME.
\end{theorem}

\paragraph{Upper limit on number of processes} 
We now consider the membership problem for the case in which there is
a fixed upper limit on the number of processes. Note that this notion
is not covered by \refasm{asm:proc-bounded}, which states that the
number of processes for any particular implementation or specification is
bounded. The results here state that if we place an upper bound on the
number of processes, with that bound being applied to all specifications and implementations
being considered,
then the set of membership problems that satisfy this bound can be solved in polynomial time. 

As before, we start by defining an FA whose language is $[\sigma]_R$.
Given some $\sigma_i$, the basic idea is that the state of the FA will be a tuple that, for each process $p$, records the most
recent event on $p$.  Thus, a state $q$ will be represented by a tuple
of events (the most recent events observed on each process) and an
event $a$ on process $p$ will only be possible in state $q$ if the
event that immediately precedes $a$ on $p$ is in this tuple.

\begin{definition}
  Suppose that \trace $\sigma_i = e_1 \ldots e_k$, in which each $e_i$
  is distinct, has projection $\pi_p(\sigma_i) = e_1^p \ldots
  e_{k_p}^p$ on process $p$ ($1 \leq p \leq n$).  We let $\machinec
  {\sigma_i}$ be the FA $(T,q_0,\Sigma,t,F)$ such
  that
  \begin{itemize}
  \item $\Sigma = \{e_1, \ldots, e_k\}$,
  \item $T = \{e_0^1,e_1^1, \ldots, e_{k_1}^1, \varepsilon\} \times
    \{e_0^2, e_1^2, \ldots, e_{k_2}^2, \varepsilon\} \times \ldots
    \times \{e_0^n,e_1^n, \ldots, e_{k_n}^n, \varepsilon\}$,
  \item $q_0 = (e_0^1, \ldots, e_0^n)$,
  \item $F = \{(e_{k_1}^1, e_{k_2}^2,
    \ldots, e_{k_n}^n)\}$, and
  \item $(T,a,T') \in t$ for $a \in \Sigma(p)$ if and only if the
    following hold: $T = (e_{j_1}^1, e_{j_2}^2, \ldots, e_{j_n}^n)$,
    $j_p < k_p$, $a = e_{j_{p}+1}^{p}$, and $T' = (e_{j_1}^1,
    e_{j_2}^2, \ldots, e_{j_{p}+1}^{p}, \ldots, e_{j_n}^n)$.
  \end{itemize}

\end{definition}

The following defines $\machinexx \sigma$.

\begin{definition}
Given \trace $\sigma = \sigma_1 \sigma_2 \ldots \sigma_k$
such that each $\sigma_i$ is end-to-end quiescent,
\[
\machinexx \sigma = \machinec {\sigma_1} \cdot \machinec {\sigma_2} \cdot \ldots \cdot \machinec {\sigma_k}.
\]
\end{definition}

\begin{lemma}
  Given \trace $\sigma$ and specification $\calS$, $\sigma \in
  \trs(\sdelta)$ iff $L(\machinexx \sigma) \cap L(S) \neq \emptyset$.
\end{lemma}

The important point now is that the state set of $\machinec{\sigma_i}$
has size that is exponential in terms of the number of processes but
if the number of processes is bounded then the size is polynomial in
terms of the length of $\sigma$.  In particular, if there is an upper
bound $b$ on the number of processes then $\machinec{\sigma_i}$ has at
most $|\sigma_i|^b$ states.  We therefore obtain the following result.

\begin{theorem}\label{theorem:poly_memb_qsc}
  If $b$ is an upper limit on the number of processes for each \trace
  of $\calQ$ then the membership problem for \qsc is in PTIME.
\end{theorem}

\section{Correctness for \qsc}
\label{sec:correctness-qsc}

This section now presents decidability and complexity results for
\qsc. Following the pattern of the previous sections, we present
unrestricted \qsc (\refsec{sec:unrestricted-qsc-1}), and then
restricted \qsc (\refsec{sec:restricted-qsc-1}).

\subsection{Unrestricted \qsc}
\label{sec:unrestricted-qsc-1}

If we consider a system with two processes $1$ and $2$ such that $e_1,
e_2 \in \Sigma(1)$ and $e_3 \in \Sigma(2)$ then we have that:
$(e_1,e_2) \in R$, $(e_2,e_1') \in R$ but $(e_1,e_2') \not \in
R$. Therefore, the independence relation is not transitive and so we
expect correctness to be undecidable (Lemma \ref{lem:rat-tr-lang}).
It could, however, be argued
that we might have changed the nature of the problem by placing
restrictions on the structure of the specification.  In this section
we therefore prove that, as expected, the correctness problem is
undecidable for \qsc.  The proof will be based on showing how an
instance of Post's Correspondence Problem can be reduced to an
instance of the correctness problem for \qsc.

\begin{definition}
Given alphabet $\Gamma$ and sequences $\alpha_1, \ldots, \alpha_n \in \Gamma^*$ and $\beta_1, \ldots, \beta_n \in \Gamma^*$,
Post's Correspondence Problem (PCP) is to decide whether there is a non-empty sequence
$i_1 \ldots i_k \in [1,n]$ of indices such that
$\alpha_{i_1} \ldots \alpha_{i_k} = \beta_{i_1} \ldots \beta_{i_k}$.
\end{definition}

Post's Correspondence Problem is known to be undecidable \cite{post46}.

First we explain how the proof operates.  Given an instance of the PCP
defined by sequences $\alpha_1, \ldots, \alpha_n$ and $\beta_1,
\ldots, \beta_n$, we will construct FA $\machinepcp{\alpha_1, \ldots,
  \alpha_n,\beta_1, \ldots, \beta_n}$ that will act as the
implementation.  The quiescent \traces of this FA will have a
particular form: a quiescent \trace $\sigma$ will be defined by a
sequence $i_1 \ldots i_k \in [1,n]$ of indices, the projection of
$\sigma$ on process $1$ will correspond to $\alpha_{i_1} \ldots
\alpha_{i_k}$ and the projection of $\sigma$ on process $2$ will
correspond to $\beta_{i_1} \ldots \beta_{i_k}$.  Thus, there is a
solution to this instance of the PCP if and only if
$\machinepcp{\alpha_1, \ldots, \alpha_n,\beta_1, \ldots, \beta_n}$ has
a non-empty quiescent trace $\sigma$ such that the projections on the
two processes define the same sequences in
$\Gamma^*$.  We then define a specification $\calS$ that allows the
set of \traces in which the projections on the two
processes differ (plus the empty sequence).  When brought
together, $\machinepcp{\alpha_1, \ldots, \alpha_n,\beta_1, \ldots,
  \beta_n}$ is not a correct implementation of $\calS$ under \qsc if
and only if $\machinepcp{\alpha_1, \ldots, \alpha_n,\beta_1, \ldots,
  \beta_n}$ has a quiescent \trace whose projections on the two
processes define the same sequences in $\Gamma^*$ and this is the case
if and only if there is a solution to this instance of the PCP.  As a
result, correctness under \qsc being undecidable follows from the PCP
being undecidable.

We now construct the FA $\machinepcp{\alpha_1, \ldots,
  \alpha_n,\beta_1, \ldots, \beta_n}$.  There will be two processes and for each $p
\in \{1,2\}$ and letter $a$ in the alphabet $\Gamma$ used, we
will create an invoke event $\pcpi{e}{a}{p}$ and a response event
$\pcpr{e}{a}{p}$.  We will add two additional events: a matching
invoke $e$ and response $\return{e}$ for process $1$. (This choice of
process $1$ for $e$ and $\rete$ is is arbitrary, i.e., we could have
chosen $e$ and $\rete$ to be events of process $2$ as well).  Thus, we
have that:
\[
\Sigma = \{\pcpi{e}{a}{p} \mid a \in \Gamma \wedge p \in \{1,2\}\} \cup
\{\pcpr{e}{a}{p} \mid a \in \Gamma \wedge p \in \{1,2\}\} \cup \{e,\rete\}
\]

We define a mapping from a sequence in $\Gamma^*$ and process number
$p$ to a sequence in $\Sigma^*$ as follows (in which $a \in \Gamma$
and $\gamma \in \Gamma^*$).
\begin{align*}
to_\Sigma(\varepsilon,p)  ={} & \varepsilon & \qquad 
to_\Sigma(a \gamma,p)  = {}& \pcpi{e}{a}{p} \:\: \pcpr{e}{a}{p} \:\: to_\Sigma(\gamma,p)
\end{align*}
Then, we define an equivalence relation on sequences in $\Sigma^*$
that essentially `ignores' the process number.
\[
\begin{array}{rcl}
  \varepsilon & \equiv & \varepsilon \\
  \pcpi{e}{a}{p} \: \sigma_1 & \equiv & \pcpi{e}{b}{q}  \: \sigma_2
  \mbox{ if and only if } a = b \mbox{ and } \sigma_1 \equiv \sigma_2
  \\ 
  \pcpr{e}{a}{p} \: \sigma_1 & \equiv & \pcpr{e}{b}{q}  \: \sigma_2
  \mbox{ if and only if } a = b \mbox{ and } \sigma_1 \equiv \sigma_2 
\end{array}
\]

From the initial state $q_0$ of $\machinepcp{\alpha_1, \ldots,
  \alpha_n,\beta_1, \ldots, \beta_n}$ there is a transition with label
$e$ to state $q$.  For all $1 \leq i \leq n$ there is a path from
$q$ to state $q'$ with the following label:
\[
to_\Sigma(\alpha_i,1) \: to_\Sigma(\beta_i,2)
\]

Similarly, for all $1 \leq i \leq n$ there is a path from $q'$ to
state $q'$ with label $to_\Sigma(\alpha_i,1) \: to_\Sigma(\beta_i,2)$.
Finally, there is a transition from $q'$ with label $\rete$ to the
unique final state $q_F$.

Now consider a path from $q$ or $q'$ that ends in $q'$.
Such a path must have a corresponding sequence 
$i_1 \ldots i_k \in [1,n]$ of indices and the projections of the label of this path
on processes $1$ and $2$ are 
$\alpha_{i_1} \ldots \alpha_{i_k}$ and $\beta_{i_1} \ldots \beta_{i_k}$ respectively.
We therefore have the following key property.

\begin{lemma}\label{lemma:PCP_P}
  $L(\machinepcp{\alpha_1, \ldots, \alpha_n,\beta_1, \ldots,
    \beta_n})$ contains a \trace $e \: \sigma \: \rete$ such that
  $\pi_1(\sigma) \equiv \pi_2(\sigma)$ iff there is a solution to the
  instance of the PCP defined by $\alpha_1, \ldots, \alpha_n$ and
  $\beta_1, \ldots, \beta_n$.
\end{lemma}

Another important property is that \traces in $L(\machinepcp{\alpha_1,
  \ldots, \alpha_n,\beta_1, \ldots, \beta_n})$ are end-to-end
quiescent.  If we define a FA $\calS$ such that $\trs(\calS)$ is the
language of all \traces of the form $\sigma \: e\: \rete$ such that
\trace $\sigma$ does not have the same projections at processes $1$
and $2$ then we will have that $\trs(\machinepcp{\alpha_1, \ldots,
  \alpha_n,\beta_1, \ldots, \beta_n}) \subseteq \trs(S)$ if and only
if there is no solution to this instance of the PCP.
The following shows how we can construct such a specification $\calS$.

$\calS$ has initial state $s_0$ and for all $a \in \Gamma$ there is a
cycle that has label $\pcpi{e}{a}{1} \:\pcpr{e}{a}{1}\:
\pcpi{e}{a}{2}\: \pcpr{e}{a}{2}$.  This cycle models the case where
the projections are identical.  We add the following to $\calS$ to
represent the ways in which a first difference in projections, on
processes $1$ and $2$, can occur.
\begin{enumerate}
\item For all $a, b \in \Gamma$ with $a \neq b$ there is a path from
  $s_0$ to state $s_1$ with label $\pcpi{e}{a}{1}\: \pcpr{e}{a}{1}\:
  \pcpi{e}{b}{2}\: \pcpr{e}{b}{2}$.  In $s_1$ there are separate cycles with labels
  $\pcpi{e}{a}{1}\: \pcpr{e}{a}{1}$ and $\pcpi{e}{a}{2}\:
  \pcpr{e}{a}{2}$ for all $a \in \Gamma$.

\item For all $a\in \Gamma$  there is a path from $s_0$ to state $s_2$
with label $\pcpi{e}{a}{1}\: \pcpr{e}{a}{1}$.
In $s_2$ there is a cycle with label
$\pcpi{e}{a}{1} \pcpr{e}{a}{1}$ for all $a \in \Gamma$.

\item For all $a\in \Gamma$  there is a path from $s_0$ to state $s_3$
with label $\pcpi{e}{a}{2}\: \pcpr{e}{a}{2}$.
In $s_3$ there is a cycle with label
$\pcpi{e}{a}{2}\: \pcpr{e}{a}{2}$ for all $a \in \Gamma$.
\end{enumerate}

State $s_1$ represents the case where after some common prefix we have
$\pcpi{e}{a}{1}\: \pcpr{e}{a}{1}$ at process $1$ and $\pcpi{e}{b}{2}\:
\pcpr{e}{b}{2}$ at process $2$ for some $a \neq b$.  States $s_2$ and $s_3$
model the cases where one projection is a proper prefix of the other
($s_2$ models the case where the projection on process $1$ is longer
and $s_3$ models the case where the projection on process $2$ is
longer).  We add a final state $s_F$ and paths from $s_1, s_2, s_3$ to
$s_F$ with label $e\, \rete$.  The following is immediate from the
construction.

\begin{lemma}\label{lemma:PCP_S}
\Trace $\sigma \in \Sigma^*$ is in $L(\calS)$ if and only if $\sigma = \sigma' \: e \: \rete$ for some $\sigma'$
such that $\pi_1(\sigma') \not \equiv \pi_2(\sigma')$.
\end{lemma}

We therefore obtain the following result.

\begin{theorem}
The correctness problem for \qsc is undecidable.
\end{theorem}

\begin{proof}
This follows from Lemmas \ref{lemma:PCP_P} and \ref{lemma:PCP_S} and the PCP being undecidable.
\end{proof}

\subsection{Restricted \qsc} 
\label{sec:restricted-qsc-1}

Our results for restricted \qsc extends the constructions used in
\refsec{sec:upper-bound-unrestr}. Given FA $\calM =
(M,m_0,\Sigma,t,M_\dagger)$ and state $m \in M$, we let
$\trsi{q}{\calM}$ denote the set of \traces that are equivalent to
end-to-end quiescent \traces that label paths that start at $m$, i.e.,
\[
\trsi{q}{\calM} = \{ \sigma \in \cls(\prei{m}{\calM}) \mid \ete{\sigma}\}
\]

We now consider the case where there is a fixed bound $b$ on the
number of events that can occur in an end-to-end quiescent trace.  For
the unbounded case we could not use Algorithm \ref{algorithm1} since
this would require us to decide whether $\trsi{q_c}{\calQ} \subseteq
\trsi{S_c}{\calS}$ and we know that this is undecidable.  However, it
is straightforward to see that in the presence of bound $b$ the
conditions controlling the loop and If statement of Algorithm
\ref{algorithm1} involve reasoning about finite languages and are
decidable.  We can therefore apply Algorithm \ref{algorithm1} and will
now show that correctness is in PSPACE.

\begin{proposition}\label{QC_checkinclusion2a}
  Let us suppose that there is a limit $b$ on the length of end-to-end
  quiescent \traces in $\calQ$ and $\calS$.  It is possible to decide
  whether $\trsi{q_c}{\calQ} \subseteq \trsi{S_c}{\calS}$ in NP.
\end{proposition}

\begin{proof}
  A non-deterministic Turing Machine can solve this problem as
  follows.  First, it guesses a \trace $\sigma$ whose length is at
  most the upper limit. 
  The Turing
  Machine then checks that $\sigma$ is end-to-end quiescent and
  whether $\sigma \in \trsi{q_c}{\calQ}$ and $\sigma \in
  \trsi{S_c}{\calS}$; from Theorem \ref{theorem:poly_memb_qsc1} we
  know that these checks can be performed in time that is polynomial
  in terms of the sizes of $\calQ$ and $\calS$.  It then returns
  failure if and only if $\sigma \in \trsi{q_c}{\calQ}$ and $\sigma
  \not \in \trsi{S_c}{\calS}$.  Thus, a non-deterministic Turing
  Machine can solve this problem in polynomial time and so the problem
  is in NP.
\end{proof}

\begin{proposition}\label{QC_checkmove2a}
  Let us suppose that there is a limit $b$ on the length of end-to-end
  quiescent \traces in $\calQ$ and $\calS$.  Given states $q_{c-1}$
  and $q_c$ of $\calQ$ and sets $S_{c-1}$ and $S_c$ of states of
  $\calS$, it is possible to decide whether there exists \trace
  $\sigma_c$ such that $\esti{q_{c-1}}{\calQ}{\sigma_c}{q_{c}}$ and
  $\estii{S_{c-1}}{\calS}{\sigma_c}{S_c}$ in NP.
\end{proposition}

\begin{proof}
  A non-deterministic Turing Machine can solve this as follows.  It
  guesses a \trace $\sigma$ whose length is at most $b$, requiring
  constant, and checks that $\sigma$ is end-to-end quiescent.  It then
  determines whether $\sigma \in \trsi{q_c}{\calQ}$ and whether
  $\sigma \in \trsi{S_c}{\calS} \setminus \trsi{S \setminus
    S_c}{\calS}$.  The non-deterministic Turing Machine returns True
  if and only if it finds that $\sigma \in \trsi{q_c}{\calQ}$, $\sigma
  \in \trsi{S_c}{\calS}$, and $\sigma \not \in \trsi{S \setminus
    S_c}{\calS}$.  From Theorem \ref{theorem:poly_memb_qsc1} we know
  that these steps can be performed in time that is polynomial in the
  sizes of $\calQ$ and $\calS$ and so in polynomial time.  Thus, a
  non-deterministic Turing Machine can solve this problem in
  polynomial time and so the problem is in NP.
%
%
\end{proof}

\begin{theorem}
The correctness problem for restricted \qc is in PSPACE.
\end{theorem}

\begin{proof}
From Propositions \ref{QC_checkinclusion2a} and \ref{QC_checkmove2a}
we know that the two conditions in Algorithm \ref{algorithm1} can
be decided in NP.
Thus,
a non-deterministic Turing Machine
can apply Algorithm \ref{algorithm1} using polynomial space.
\end{proof}

%
%
%
%

\newcommand\Tstrut{\rule{0pt}{2.6ex}}         
\newcommand\Bstrut{\rule[-0.9ex]{0pt}{0pt}}   
\section{Conclusions}
\label{sec:conclusions}

Concurrent objects (such as the queue example in
\refsec{sec:quiesc-cons-queue}) form an important class of objects,
managing thread synchronisation on behalf of a programmer.  The safety
properties that a concurrent object satisfies can be understood in
terms of the correctness conditions such as sequential consistency,
linearizability and quiescent consistency \cite{HeSh08}. 

Decidability and complexity for checking membership and correctness
for these conditions have been widely studied. These generally extend
Alur et al.'s methods \cite{AlurMP00}, which in turn is based on the
notions of independence from Mazurkiewicz Trace Theory. A summary of
results from the literature is given in \reftab{tab:summary}. The
bounded and unbounded versions of linearizability that have been
studied refer to the number of processes that a system is assumed to have
--- the unbounded version does not assume an upper limit on the number
of processes. Our results on \qc and \qsc adds to this existing body
of work.

\begin{table}[h]
  \begin{minipage}[t]{\linewidth}\small
    \begin{tabular}[t]{p{4cm}|l|l|p{2.6cm}|p{2.4cm}}
      & \multicolumn{2}{c|}{\bf Membership} &
      \multicolumn{2}{c}{\bf Correctness} \\
      \cline{2-5}
      \textbf{Correctness condition} & Unrestricted & Restricted      &
      Unrestricted & Restricted
      \\
      \hline
      Sequential consistency & NP-complete \cite{GibbonsK92} & --- & 
      Undecidable \cite{AlurMP00} & ---
      \Tstrut\\[1mm]
      Linearizability      & NP-complete${}^{(a)}$\cite{GibbonsK92} &
      PTIME${}^{(b)} $\cite{GibbonsK92} & EXPSPACE${}^{(c)}$
      \cite{AlurMP00}, & EXPSPACE-
      \\
      & & & Undecidable${}^{(b)}$~\cite{BouajjaniEEH13} &
      complete${}^{(d)}$ \cite{BouajjaniEEH13}
      \\[1mm]
      Serializability & ---         & ---      & PSPACE
      \cite{AlurMP00} & --- 
      \\[1mm]
      Conflict serializability & ---         & ---      & EXPSPACE-complete
      \cite{BouajjaniEEH13} & --- 
      \\[5mm]
      \Qc${}^\star$ & NP-complete & PTIME ${}^{(b)}$ & coNP-hard, 
      & PSPACE ${}^{(b)}$ \\
      &             &       & EXPSPACE
      \\[1mm]
      Quiescent sequential & NP-complete & PTIME ${}^{(b) \text{ or } (e)}$ &
      Undecidable & PSPACE ${}^{(b)}$\\
      consistency ${}^\star$ & & & 
    \end{tabular}
    \footnotetext{${}^{(a)}$ Finite number of processes}
    \footnotetext{${}^{(b)}$ Predetermined upper limit on number of processes}
    \footnotetext{${}^{(c)}$ Potentially infinite number of processes}

    \footnotetext{${}^{(d)}$ Implementations with fixed linearization
      points, but potentially infinite number of processes}

    \footnotetext{${}^{(e)}$ Upper limit on number of events between
      two quiescent states}

    \footnotetext{${}^{\star}$ Our results}

    \caption{Summary of decidability and complexity results}
    \label{tab:summary}
    
  \end{minipage}
\end{table}

The notion of \qc we have considered is based on the definition by
Derrick et al.\ \cite{DerrickDSTTW14}, which is a formalisation of the
definition by Shavit \cite{Shavit11}. This definition allows operation
calls by the same process to be reordered, i.e., sequential
consistency is not necessarily preserved. However, for concurrent
objects, sequential consistency is known to be necessary for
observational refinement \cite{Filipovic-LinvsRef2010}, which in turn
guarantees substitutability on behalf of a programmer. Therefore, we
also study a stronger version \qsc that disallows commutations of
events corresponding to the same process.



There are further variations of \qc in the literature. Jagadeesan and
Riely have developed a quantitative version of \qc \cite{JagadeesanR14}
that only allows reordering if there is adequate contention in the
system; here adequate contention is judged in terms of the number of
open method calls in the system. It is straightforward to extend the approach used in this paper
to show that the membership problem for this quantitative version is NP-complete,
but decidability of correctness is not yet known. Versions of \qc
su ited to relaxed-memory architectures have also been developed
\cite{SmithDD14,DongolDGS15}, where the notion of a quiescent state
incorporates pending write operations stored in local
buffers. Consideration of decidability and complexity for membership
and correctness of these different variations is a task for future
work. In particular, Jagadeesan and Riely's condition forms forms a
class of quantitative correctness conditions, which includes
quantitative relaxations of linearizability
\cite{AfekKY10,HenzingerKPSS13} and sequential consistency
\cite{Sezgin15}.

Bouajjani et al.\ have developed characterisations of algorithm
designs that enable reduction of the linearizability problem to a
(simpler) state reachability problem \cite{BouajjaniEEH15}. Other work
\cite{BouajjaniEEH15-POPL} has considered (under) approximations of
history inclusion with the aim of solving the \emph{observational
  refinement} problem for concurrent objects
\cite{HeHS86,Filipovic-LinvsRef2010} directly. Linking \qc and \qsc to
the state reachability problem and under approximations for 
observational refinement are both topics for future work.

\paragraph{Acknowledgements.} We thank Ahmed Bouajjani, Constatin
Emea and Gustavo Petri for helpful discussions.

\def\bibsection{\section*{References}}





\bibliographystyle{plain}
\bibliography{correctness,references}

\end{document}